\documentclass[11pt,twoside,reqno]{atmp}

\usepackage{amsmath,amssymb,amsthm,eucal,mathrsfs,yfonts,bbm,bbm}
\usepackage[T1]{fontenc}

\newcommand{\g}{{g}} 
\newcommand{\Mk}{\mathbbmss{M}} 

\newcommand{\dc}{\mathcal{O}} 

\newcommand{\A}{\mathcal{A}} 

\newcommand{\al}{\alpha}

\newcommand{\Ga}{\Gamma}
\newcommand{\de}{\delta}
\newcommand{\De}{\Delta}
\newcommand{\eps}{\varepsilon}

\newcommand{\la}{\lambda}
\newcommand{\La}{\Lambda}
\newcommand{\si}{\sigma}
\newcommand{\ph}{\varphi}

\newcommand{\sst}[1]{\scriptscriptstyle{#1}}

\newcommand{\id}{\mathrm{id}} 

\newcommand{\NN}{\mathbbmss{N}} 
\newcommand{\R}{\mathbbmss{R}} 
\newcommand{\C}{\mathbbmss{C}} 
\newcommand{\ZZ}{\mathbbmss{Z}}
\newcommand{\supp}{\mathrm{supp}}
\newcommand{\WF}{\mathrm{WF}} 

\newcommand{\Bcal}{\mathcal {B}}
  
\newcommand{\Ccal}{\mathcal{E}}
\newcommand{\Dcal}{\mathcal{D}}
\newcommand{\Ecal}{\mathcal{E}} 
\newcommand{\Fcal}{\mathcal{F}}

\newcommand{\Tcal}{\mathcal{T}}

\copyrightnotice{2009}{0}{0}{0} 

\begin{document}

\title[RG in Algebraic Quantum Field Theory]
{Perturbative Algebraic Quantum Field Theory and the Renormalization Groups}

\arxurl{0901.2038 [math-ph]} 

\author[Brunetti, D\"utsch, Fredenhagen]{R. Brunetti$^{(1)}$, M. D\"utsch$^{(2)}$,  K. Fredenhagen$^{(3)}$}

\address{$^{(1)}$ Dipartimento di Matematica, Universit\`a di Trento,\\
Via Sommarive 14, I-38050 Povo (TN), Italy\\
and I.N.F.N. sez. Trento\\
$^{(2)}$  Courant Research Center "Higher Order Structures\\ 
in Mathematics", University G\"ottingen,\\
Mathematisches Institut, Bunsenstr.~3-5, \\
D-37073 G\"ottingen, Germany\\
$^{(3)}$ II Inst. f. Theoretische Physik, Universit\"at Hamburg,\\
Luruper Chaussee 149, D-22761 Hamburg, Germany}  
\addressemail{brunetti@science.unitn.it, duetsch@physik.unizh.ch, klaus.fredenhagen@desy.de}

\begin{abstract}
A new formalism for the perturbative construction of algebraic quantum field theory is developed. The formalism allows the  treatment of low dimensional theories and of non-polynomial interactions. We discuss the connection between the St\"uckelberg-Petermann renormalization group which describes the freedom in the perturbative construction with the Wilsonian idea of theories at different scales . In particular we relate the approach to renormalization in terms of Polchinski's Flow Equation to the Epstein-Glaser method. We also show that the renormalization group in the sense of Gell-Mann-Low (which characterizes 
the behaviour of the theory under the change of all scales) is a 1-parametric subfamily of the St\"uckelberg-Petermann group and that this subfamily 
is in general only a cocycle. Since the algebraic structure of the  St\"uckelberg-Petermann group does not depend on global quantities,
this group can be formulated in the (algebraic) adiabatic limit without meeting any infrared divergencies. In particular we derive an 
algebraic version of the Callan-Symanzik equation and define the $\beta$-function in a state independent way. 
\end{abstract}

\maketitle

\newpage

  \theoremstyle{plain}
  \newtheorem{df}{Definition}[section]
  \newtheorem{teo}[df]{Theorem}
  \newtheorem{prop}[df]{Proposition}
  \newtheorem{cor}[df]{Corollary}
  \newtheorem{lemma}[df]{Lemma}
  
  \theoremstyle{plain}
  \newtheorem*{Main}{Main Theorem}
  \newtheorem*{MainT}{Main Technical Theorem}

  \theoremstyle{definition}
  \newtheorem{oss}[df]{Remark}

 \theoremstyle{definition}
 \newtheorem{ass}{\underline{\textit{Assumption}}}[section]

 \tableofcontents

\section{Motivations and plan}
The locally covariant formulation of quantum field theory \cite{BFV,HW1} is based on the principle that the theory has to be built in terms of quantities that are uniquely determined by local properties of spacetime. This means in particular that concepts like vacuum states and particles should not enter the
formulation of the theory. Instead these concepts are relevant in the interpretation of the theory under suitable circumstances.

Historically, the main obstacle in performing a purely local construction was the important r\^{o}le the spectrum condition plays for finer properties of the theory, and it needed the insight of Radzikowski \cite{Rad} that the local information of the spectrum condition can be formulated in terms of a condition on the wave front sets (microlocal spectrum condition, see also \cite{BFK}). After this breakthrough, a thorough construction of renormalized perturbative quantum field theory on generic curved spacetimes could successfully be performed \cite{BF,HW1,HW2} in the framework of Algebraic Quantum Field Theory.

Nevertheless, the theory obtained looks still somewhat remote from more standard formulations of quantum field theory as one may find them in typical text books, and one may ask how ideas like the Renormalization Group show up in the locally covariant framework. A first answer to this latter question was already given 
in \cite{HW3}  (for a more explicit formulation in Minkowski space see \cite{DFret} and \cite{DFbros}), 
and an elaboration on this approach was the starting point of the present work. 

It turned out to be appropriate to revise the formulation in such a way that the dependence on
nonlocal features is eliminated completely. In older formulations a dependence on the choice of a Hadamard 2-point function (replacing the nongeneric vacuum) was used, and it had to be shown that the theory is actually independent of this choice. The present point of view is that the Hadamard functions are only used in order to characterize a topology. The topology is then shown to be independent of the Hadamard function used. All algebraic structures, however, are defined without recourse to a Hadamard function.

The plan of the paper is the following: The next section gives definitions and constructions of the main technical elements of our procedures, namely all kind of algebraic structures, as time-ordered products, for instance. It contains a simplification of the rigorous treatment of our framework which is postponed to the next section. A particularly pressed reader may skip all technicalities of the $3$d section and concentrate only on the main definitions of the Hadamard families of functions and local observables. The $4$th section deals with the St\"uckelberg-Bogoliubov-Epstein-Glaser approach to renormalization, but done in a novel manner. 
The next section is the core of our paper. It contains the definitions and comparison of three different versions of the renormalization group. The discussion is continued in the $6$th section where the algebraic adiabatic limit is discussed, via the introduction of a novel tool in algebraic quantum field theory termed generalized Lagrangians, and an algebraic form of the Callan-Symanzik equation is found. The last section deals with a couple of typical examples as $\varphi^4$ in $4$ dimensions, and $\varphi^3$ in $6$ dimensions. Here we compare with the textbook treatments (beta function) and find perfect agreement. Some technicalities are deferred to the appendices A, B and C.

\section{Definitions, initial constructions and outlook}\label{sec:defcon}
We look at the theory of a scalar field $\varphi$. 
For each spacetime $M$ of dimension $d\ge 2$ we consider the space of infinitely differentiable functions as the configuration space $\Ccal(M)\equiv C^\infty(M)$ (a Fr\'echet space). The observables of the theory are functionals on $\Ccal(M)$ which are infinitely often differentiable, such that the functional derivatives 
are test functions with compact support. A simple example is the functional $\Ccal(M) \to \C$, $\varphi\mapsto \int dx\varphi(x)f(x)$ with some test function $f\in\Dcal(M)\equiv C^\infty_0(M)$. The space of observables will be denoted by $\Fcal_0(M)$. We recall that the notion of differentiability on Fr\'echet spaces is well developed, a particularly nice introduction of which can be found in \cite{Ham}. In practice, in our treatment only the validity of the chain rule is important. Notationally, we shall use indifferently the following possibilities as equivalent writings of generic $n$th order derivatives
\[
D^n F(\varphi)(v^{\otimes n})\equiv \dfrac{\delta^n F}{\delta\varphi^n}(\varphi)(v^{\otimes n})\equiv  \left. \dfrac{d^n}{d\lambda^n}F(\varphi +\lambda v)\right |_{\lambda=0} \equiv \langle F^{(n)}(\varphi),v^{\otimes n}\rangle\ , 
\]
where $v\in\Ccal(M)$,
with the brackets denoting, from now on, either integration in the appropriate spaces, or sometimes the duality pairing of locally convex topological spaces, and the context should hopefully make precise as to which is which. Similar equivalent formulations can be adopted for the ``kernels''
\[
F^{(n)}(\varphi)(x_1,\dots,x_n)\equiv \dfrac{\delta^n F}{\delta\varphi^n}(\varphi)(x_1,\dots,x_n)\equiv\dfrac{\delta^n F}{\delta\varphi(x_1)\cdots\delta\varphi(x_n)}\ .
\]
We define the support of a functional $F\in\Fcal_0(M)$ as the set of points $x\in M$ such that $F$ depends on the behaviour of $\ph$ in every neighbourhood $U\in\mathcal{U}(x)$ of $x$,  
\begin{gather}\label{eq:support}
\supp(F)\doteq\{x\in M\, |\, \forall U\in\mathcal{U}(x)\exists \ph,\psi\in\Ecal(M),\ \supp\psi\subset U\notag\\
\text{ such that }F(\ph+\psi)\neq F(\ph)\}\ ,
\end{gather}
and we require it to be a compact set in $M$.

A virtue of the new approach is that we do not restrict ourselves to polynomial functionals. This is motivated, for example, by quantum gravity 
and by the existence of renormalizable models with non-polynomial interactions in $d=2$ dimensions.

We consider a linear hyperbolic differential operator for the scalar field 
with respect to the spacetime metric $g$ with signature $(+,-,\cdots,-)$,
\begin{equation}
P=\square_{\g}+m^2+\xi R \ ,\label{KGop}
\end{equation}
with the scalar curvature $R$ and real parameters $m^2$ and $\xi$. On a globally hyperbolic time oriented spacetime $P$ possesses unique retarded and advanced Green's functions $\De_R$ and 
$\De_A$, respectively. In terms of these Green's functions we can introduce two important structures on the space of formal power series in $\hbar$ with coefficients in the vector space $\Fcal_0(M)$. 

First we introduce a $\star$-product by a twist induced by the commutator function $\De=\De_R-\De_A$ (cf.~e.g.~\cite{DFsiena}),

\begin{equation}\label{star product}
F\star G\doteq \mathcal{M}\circ \exp({i\hbar \Gamma_\Delta})(F\otimes G) \ ,
\end{equation}
where $\mathcal{M}$ denotes pointwise multiplication,
\begin{equation}
\mathcal{M}(F\otimes G)(\varphi)\doteq F(\varphi)G(\varphi)\equiv (F\cdot G)(\varphi) \ ,
\end{equation}
and $\Gamma_\Delta$ is the functional differential operator
\begin{equation}
\Gamma_\Delta\doteq\frac{1}{2}
                  \int dx\, dy \De(x,y)\frac{\de}{\de\varphi(x)}\otimes\frac{\de}{\de\varphi(y)} \ .
\end{equation}
In this way we obtain an associative algebra where the fields satisfy the commutation relation
\begin{equation}
[\varphi(f),\varphi(g)]_{\star}=i\hbar \langle f,\De g\rangle \ , \quad f,g\in \Dcal(M)\ .
\end{equation}
Complex conjugation endows the algebra with an involution,
\begin{equation}
\overline{F\star G}=\overline{G}\star\overline{F}
\end{equation}  
since $\De$ is antisymmetric and real.

The algebra contains an ideal generated by elements of the form $\varphi(Pf)$, $f\in \Dcal(M)$. The quotient algebra is just the standard algebra of the free scalar field. It will turn out, however, that it is convenient to work with the original algebra $(\Fcal_0(M),\star)$ (off shell formalism).

The second structure we need is the time ordering operator $T$ associated with the $\star$-product. 
It is defined in terms of 
the Dirac propagator $\De_D=\tfrac{1}{2}\, (\De_R+\De_A)$ which was introduced by Dirac in his treatment of the classical interaction between a point charge and the electromagnetic field \cite{Dirac}. We set
\begin{equation}
TF\doteq\exp({i\hbar \Ga_{\De_D}})F
\end{equation}
with 
\begin{equation}
\Ga_{\De_D}\doteq\frac12\int dx\,dy \De_D(x,y)\frac{\de^2}{\de\varphi(x)\de\varphi(y)} \ .
\end{equation}
Formally, $T$ may be understood as the operator of convolution with the oscillating 
Gaussian measure with covariance $i\hbar\Delta_D$,
\begin{equation}\label{path integral}
TF(\varphi)=\int d\mu_{i\hbar\Delta_D}(\phi)F(\varphi-\phi) \ . 
\end{equation}
Its inverse $T^{-1}$ (the anti time ordering operator $\overline{T}$) is obtained by replacing $\De_D$ by $-\De_D$. It coincides, due to the reality of $\De_D$, with the complex conjugate operator
\begin{equation}
\overline{T}F=\overline{T\overline{F}}\ .
\end{equation} 

We then define the time ordered product $\mathcal{M}_T$ as the pointwise product  $\mathcal{M}$ transformed by $T$,
\begin{equation}\label{top}
\mathcal{M}_T\doteq T\circ \mathcal{M}\circ (T^{-1}\otimes T^{-1})
\end{equation}
or, in a more standard notation as a binary composition (where $\cdot$ denotes the pointwise product and $\cdot_T$ the time ordered product) 
\begin{equation}\label{eq:time} 
F\cdot_T G\doteq T(T^{-1}F \cdot T^{-1}G) \ .
\end{equation}
By comparing
\begin{equation}\label{eq:timep}
\varphi(x)\cdot_T\varphi(y)=\varphi(x)\cdot\varphi(y) + \frac{i\hbar}{2}(\De_R+\De_A)(x,y) 
\end{equation}
with
\begin{equation}\label{eq:starp}
\varphi(x)\star\varphi(y)=\varphi(x)\cdot\varphi(y) + \frac{i\hbar}{2}(\De_R-\De_A)(x,y) \ ,
\end{equation}
and using the support properties  of the Green's functions, we see that $\cdot_T$ is indeed the time ordered product with respect to the $\star$-product \eqref{star product}. 

We emphasize that the time ordered product is a well defined, associative and commutative product. For this it is important, as stressed before, not to pass to the quotient algebra, where the validity of the free field equation would be in conflict with the fact that the Dirac propagator does not solve the homogeneous Klein-Gordon equation. Indeed, the ideal generated by the field equation with respect to the $\star$-product (or, equivalently, with the pointwise product) is not an ideal with respect to the time ordered product, as may be seen from the Dyson-Schwinger type equation
\begin{equation}
F\cdot_T\ph(Pf)=F\ph(Pf) +i\hbar\langle F^{(1)},f\rangle\ .
\end{equation}

Note that the appearance of the Dirac propagator in the time ordered product is due to the choice of the commutator function $\Delta$ in the $\star$-product. If we had chosen the 
Wightman two-point function $\Delta_+=\tfrac{i}{2}\,\Delta+\Delta_1$ 
in the formula \eqref{star product} for the 
$\star$-product we would have obtained the Feynman propagator $\Delta_F=i\Delta_D+\Delta_1$ in the time ordered product. But these conventions use a notion of positive frequency which is absent on generic Lorentzian spacetimes. Moreover, the logarithmic singularities at $m^2=0$ obscure the scaling behaviour. 

We now have the means to introduce interactions in our framework. Namely, let $V$ be an arbitrary element of $\Fcal_0(M)$. Then
we define the formal $S$-matrix as the time ordered exponential
\begin{equation}\label{S-matrix}
S(V)\doteq T\circ \exp\circ\ T^{-1}(V)\equiv \exp_{\cdot_T}(V)\ .
\end{equation}
One may exhibit a factor $ig/\hbar$ in $V$. The $S$-matrix would then be a formal power series in the coupling constant $g$ and a Laurent series in 
$\hbar$ (see e.g.~\cite{DFloop}). We found it more convenient to incorporate such a factor into $V$, 
such that all expressions are formal power series in $\hbar$. The coefficients in every order in $\hbar$ are analytic functionals in $V$ which may be described in terms of their (convergent) power series expressions.
 
The nonlinear interactions $V\in \Fcal_0(M)$ are, in general, nonlocal. 
As a consequence the $S$-matrix defined above fails, in general, to be unitary for imaginary
interactions $V$. Once extended to more singular functionals, unitarity of the $S$-matrix can be restored for local interactions where it takes the form of a renormalization condition. 
Inspite of the non-unitarity, the $S$-matrix for non-local interactions is quite frequently used in quantum field theory, 
in particular when dealing with effective interactions as they appear e.g.~in the renormalization method of the 
flow equation (Section~\ref{sec:flow}).

Notice also that the presence of the inverse of time ordering in the formula \eqref{eq:time}, which would be absent in a path integral formulation on the basis of \eqref{path integral},  
remove the so-called tadpole terms.

We now want to extend the operations to more singular functionals, including in particular local functionals, characterized by the condition that their functional derivatives have compact support on the diagonal,
\begin{equation}\label{local}
\frac{\de^n F}{\de\varphi^n}(x_1,\dots, x_n) =0\ ,\quad \text{if } x_i\neq x_j \text{ for some pair } (i,j) \ ,
\end{equation}
and that their wavefront set is transversal to the tangent space of the diagonal (this may be understood as a microlocal version of translation invariance).

For the $\star$-product this can be done in the following way. Namely we choose a so-called Hadamard solution $H$ of the Klein-Gordon equation (actually, a symmetric one, see next Section) and transform the $\star$-product by the operation 
\begin{equation}
\al_H\doteq\exp({\hbar\Gamma_H}) \ ,\quad
\Gamma_H\doteq\frac12 \int dx\,dyH(x,y)\frac{\de^2}{\de\varphi(x)\de\varphi(y)} \ ,\label{def(alpha_H)}
\end{equation}
to an equivalent $\star$-product $\star_H$, defined by
\begin{equation}
F\star_H G\doteq\al_H(\al_H^{-1}(F)\star \al_H^{-1}(G)) \ .
\end{equation}
The new $\star$-product $\star_H$ is obtained from the original one by replacing $\frac{i}{2}\Delta$ by $\frac{i}{2}\Delta +H$. 

By the microlocal spectrum condition \cite{BFK,Rad} the wave front set of $\frac{i}{2}\Delta +H$ is such that the transformed product can 
now be uniquely extended by sequential continuity to a space $\Fcal(M)$ of
functionals whose derivatives are distributions with appropriate wave front sets. 
The relevant topology is the H\"ormander topology for all derivatives (see Section \ref{sec:local}).
One then defines the topology on $\Fcal_0(M)$ as the initial topology with respect 
to $\al_H$ and proves that this topology does not depend on the choice of the Hadamard 
solution $H$. The sequential completion $\A(M)$ of $\Fcal_0(M)$
can be equipped with a unique sequentially continuous $\star$-product. 
In other words, $\al_H^{-1}:\Fcal(M)\rightarrow \A (M)$ is a linear bijection and the 
$\star$-product in $\A (M)$ is defined by
$\al_H^{-1}(F)\star \al_H^{-1}(G)\doteq\al_H^{-1}(F\star_H G)$.  
Roughly speaking, the used topology is characterized by the property that 
the point splitting approximations
to nonlinear local fields, e.g.
\begin{equation}
\ph(x)\ph(y)-\hbar H(x,y)\equiv \al_H^{-1}(\ph(x)\ph(y)) \ ,
\end{equation}
converge in the coincidence limit $y\to x$ for all Hadamard functions $H$. 
In the next section this procedure will be described in more detail.

The time ordered product, however, is not continuous in the topology described 
above, as may be seen, e.g., from the fact that the powers of the Feynman like propagators $H_F\doteq i\Delta_D+H$ cannot be defined by using H\"ormander's criterion for the existence of products of distributions. Its (partial) extension amounts to the process of renormalization. Quite different recipes have been developed. These are
\begin{itemize}

\item BPHZ renormalization. It relies on an expansion of the $S$-matrix in terms of Feynman graphs, 
where the vertices with $n$ adjacent lines correspond to the $n$th functional derivatives of the potential and the lines to the Feynman propagator. 
On Minkowski space, the corresponding expression can be written as an integral on momentum space. By a clever, somewhat involved procedure (the famous Forest Formula of Zimmermann \cite{Zimmermann2}) the integrand is modified by subtraction of polynomials in the momenta. 
Recently, it was observed by Kreimer \cite{Kreimer} that this procedure may be understood as an antipode of a suitable Hopf algebra. A rigorous discussion can most easily be performed on an euclidean space (with the Feynman propagator replaced by the Green's function of the corresponding elliptic operator). A rigorous discussion on Minkowski space \cite{Zimmermann1} is somewhat involved due to the fact that the modified Feynman integrals do not converge absolutely. An extension to generic spacetimes has been tried but, to the best of our knowledge, did not yet lead to a complete construction.

\item Flow equation. The idea here is to interpolate between the pointwise and the time ordered product 
by introducing a cutoff $\Lambda$ and to study the flow of the effective potential (as a function of $\Lambda$)
in the sense of Wilson. Namely let $T_{\La}=\al_{h_\La -H}$, with a differentiable family of symmetric smooth functions $(h_{\La})_{\La\in\R_+}$ with $h_0=0$  (hence $\al_HT_0=\id$) and $h_{\La}\to H_F$ in the appropriate sense (see Section \ref{sec:local} and Subsection \ref{sec:flow}) as $\La\to\infty$, hence $\al_HT_{\La}\to \al_{H_F}$. 
By means of this family we define the regularized $S$-matrix as $S_\La= \exp_{\cdot_{T_\La}}$, 
which is invertible on the functionals we consider: $S_\La^{-1}=T_{\La}\circ\mathrm{log}\circ T_{\La}^{-1}$.
The interpolating family $V_\La$ of effective potentials at scales $\La$ is defined by the requirement that the cutoff theory
with interaction $V_\La$ is equal to the exact theory with the original local interaction $V$, $S_\La(V_\La)=S(V)$, or
explicitly
\begin{equation}
V_\La=S_{\La}^{-1}\circ S(V) \ .
\end{equation}
The effective potential $V_{\La}$ is generically non local and satisfies the flow equation 
\begin{equation}
\frac{d}{d\La}V_\La= -\frac{1}{2}\left (\frac{d}{d\La}\mathcal{M}_{T_\La}\right )(V_\La \otimes V_\La)\ ,
\end{equation}
with the interpolating time ordered products
\begin{equation}
\mathcal{M}_{T_\La}\doteq T_\La \circ \mathcal{M} \circ (T_{\La}^{-1}\otimes T_{\La}^{-1}) \ .
\end{equation}
This is  Polchinski's flow equation \cite{Polchinski} in the Wick ordered form \cite{Salmhofer1} (For a proof in our formalism see Section \ref{sec:flow}). 
Up to now it was almost exclusively used for euclidean field theory on euclidean space, where the approximate time ordering operation can be built by a momentum cutoff (see, e.g. \cite{Salmhofer2}; for Minkowski space see \cite{Kopper}). In principle there is no obstacle to perform the same construction on generic spacetimes, the only obstruction being that there seems to be no locally covariant choice of the cutoff which leads to the removal of all singularities. A partial removal can be obtained by Pauli-Villars regularization. 

\item Causal Perturbation Theory. This approach was developed by Epstein and Glaser \cite{EG} on the basis of ideas of St\"uckelberg \cite{Stuck} and Bogoliubov \cite{BS}. It applies to local functionals $V$.
It fully exploits the locality  properties of the interactions and is ideally suited for an extension to generic spacetimes \cite{BF,HW1}. Its basic idea is that the time ordered product of $n$ local functionals is, without any renormalization, up to (finite) local counter terms already determined by the time ordered product of less than $n$ local functionals. The freedom in the choice of local counter terms is exactly the freedom in the choice of renormalization conditions. No cutoff is needed in this approach. The renormalization group in the sense of St\"uckelberg-Petermann characterizes the freedom in the choice of time ordering prescriptions \cite{DFret}. 

\end{itemize}
  
It is the aim of the present paper to clarify the relation between the causal approach and the flow equation and in particular to analyze the different concepts which are denoted as renormalization group. It will turn out that one has to distinguish at least three different versions of the renormalization group:
\begin{itemize}
\item The renormalization group in the sense of St\"uckelberg-Petermann,
\item The renormalization group in the sense of Gell-Mann-Low,
\item The renormalization group in the sense of Wilson.
\end{itemize}
The renormalization group in the sense of St\"uckelberg-Petermann is formed by the family of all finite renormalizations. It is really a group. The renormalization group in the sense of Gell-Mann-Low characterizes the behaviour of the theory under the change of all scales. It is a group only in the massless case; in the massive case it is rather a cocycle. The renormalization group in the sense of Wilson refers to the dependence of the theory on a cutoff. It has no simple algebraic properties, but can be characterized in terms of Polchinski's flow equation. The relation to the Connes-Kreimer approach \cite{Connes1,Connes2} will be postponed to a future paper.
  
Since only the causal approach has been extended to generic spacetimes we restrict our treatment in the following to Minkowski space $\Mk$.
Even there, the proper treatment of the dependence on the mass term is not trivial. As the locally covariant approach suggests, all real values of the parameter $m^2$ should be allowed, in spite of the fact that a vacuum state can exist only for nonnegative values of $m^2$. Since the Green's functions of the Klein Gordon operator for $m^2<0$ are no longer tempered distributions (see e.g. \cite{Bert}), a discussion in terms of support properties in momentum space  is, in general, not possible. Here, the methods of Microlocal Analysis \cite{Hor} are particularly fruitful.

\section{Enlargement of the space of observables}\label{sec:local}
In order to include nontrivial local interactions we have to enlarge the space
$\Fcal_0(\Mk)$. We do this  by transforming the 
$\star$-product (\ref{star product}) into an equivalent one corresponding to normal ordering. The standard way of doing this is to use the transformation
\[
\al_{\De_1}\doteq\exp({\hbar\Ga_{\De_1}}) : 
\Fcal_0(\Mk)\rightarrow \Fcal_0(\Mk)
\]
with 
\begin{equation}
\Ga_{\De_1}\doteq\frac{1}{2} \int dx dy\, \De_1(x,y) \frac{\de^2}{\de\varphi(x)\de\varphi(y)}
\end{equation}
where $\De_1$ is the symmetric part of the 2-point function, and to set
\begin{equation}
F\star_{\De_1}G\doteq\al_{\De_1}(\al_{\De_1}^{-1}(F)\star \al_{\De_1}^{-1}(G))\ .
\end{equation}
This just produces the standard Wick ordering. 
It has the nice feature that the product can now be extended to more general functionals, in particular to composite fields, smeared with test functions. 

The disadvantage of this prescription is that, as a function of $m^2$, it is not smooth at $m^2=0$ 
(and is not defined at $m^2<0$ ($\le 0$ in 2 dimensions)). In the light of a generally covariant framework this is problematic. In particular the smooth behavior under scaling of all dimensionful parameters at zero was crucial for the renormalization method of \cite{HW2,HW3}. 

We therefore replace $(\De_{1\,m})_{m^2> 0}$
by a family of symmetric distributions  (Hadamard functions)
$H=(H_m)_{m^2\in\R}$, $H_m\in\Dcal^\prime(\Mk^2)$, such that

\begin{itemize}

\item $H_m$ is a distributional solution of the Klein-Gordon equation in both arguments;

\item $H_m$ is invariant under Poincar\'{e} transformations;

\item $H_m+i\De_m/2$ satisfies the microlocal spectrum condition \cite{Rad,BFK};

\item For each test function $f\in \Dcal(\Mk^2)$, $\langle H_m,f \rangle$ is a smooth function of $m^2$;

\item $H_m$ scales almost homogeneously, i.e.  $ \varrho^{d-2}H_{m/\varrho}(\varrho x,\varrho y)$ is a poly\-no\-mial in $\log\varrho$.
 
 \end{itemize}

For $m^2> 0$, $H_m$ differs from $\De_{1\,m}$ by a smooth Poincar\'{e} invariant bisolution of the 
Klein Gordon equation. There is a crucial difference in the scaling behavior of $H\equiv H^{(d)}$ for even and odd
dimensions $d$ of Minkowski space \cite{DFret}. 
\begin{itemize}

\item For $d$ odd $H$ is uniquely determined (and thus scales even homogeneously). 

\item In even dimensions homogeneous scaling is not compatible with 
smoothness in $m^2$. $H$ is not uniquely determined by the conditions above, but depends 
on an additional mass parameter $\mu>0$. 
One defines
\begin{equation}
v(x,y)\doteq \frac{1}{2}\,\mu\frac{\partial}{\partial\mu} H^{\mu}_m(x,y)\label{def(v)}
\end{equation}
which can be proved to be a smooth function.

In the standard literature smoothness in $m^2$ is not required; but even with that the Wightman 
2-point function cannot be used in the massless 2-dimensional theory, since 
$\Delta^{+\,(2)}_m$ is logarithmic divergent for $m\to 0$.
\end{itemize}
Derivations of the explicit expressions for $H$ and $v$
proving these statements are given in Appendix \ref{Hadamard}.

\subsection{Algebras of observables and smooth dependence on $m^2$}
The linear maps $\al_{H^{\mu}}$ deform for every value of $m^2$ the $\star$-products $\star_m$ into 
equivalent products $\star_{m,\mu}$ which are smooth in $m^2$, in the sense that 
\begin{equation}\label{smoothness in m}
\R\ni m^2\longrightarrow (F\star_{m,\mu}G)(\varphi)
\end{equation}
is smooth for all $F,G \in \Fcal_0(\Mk)$ and all $\varphi\in\Ccal(\Mk)$. 
 
We now enlarge the space of functionals $\Fcal_0(\Mk)$ to the space $\Fcal(\Mk)$ of 
functionals which are infinitely differentiable, such that the $n$th functional derivatives are distributions with compact support and wavefront sets in the following subset of the cotangent bundle of $\Mk^n$,
\begin{equation}
\Xi_n=\{(x_1,\dots, x_n, k_1,\dots, k_n)\ |\  (k_1,\ldots k_n)\not\in (\overline{V}_+^{\, n}\cup \overline{V}_-^{\, n})\} \ .
\end{equation}
We equip this space with the following topology. First we endow the space of distributions with  wavefront sets contained in $\Xi_n$ with the H\"ormander  topology. We then define the topology on $\Fcal(\Mk)$ as the initial topology for the maps 
\begin{equation}
F\longrightarrow\frac{\de^n F}{\de\ph^n}(\ph)\ , \quad n\in\NN_0 \ .
\end{equation}
The H\"ormander topology for the space $\mathcal{E}'_C(\Mk)$ of compactly supported distributions  $t$ with wavefront sets in a closed cone  $C$ in the cotangent space is defined in the following way. 
By the definition of the wave front set (see e.g. \cite{Hor}) every properly supported pseudodifferential operator $A$ containing the set $C$ in its characteristic set will map the distribution
$t$ into a smooth function.
The topology on 
$\mathcal{E}'_C(\Mk)$ is now the initial topology for the maps 
\begin{align}
t&\longrightarrow\langle t, f\rangle\ ,\quad f\in\Ecal(\Mk)\ ,
\label{topology1}\\
t&\longrightarrow At\in\Ecal(\Mk)\label{topology2}\ ,
\end{align}
for all  properly supported pseudodifferential operators $A$ with characteristics containing $C$.
The H\"ormander topology for an open cone as $\Xi_n$ is defined as the inductive limit for all closed cones contained in it.
  
With respect to this topology, $\Fcal(\Mk)$ is the sequential completion of $\Fcal_0(\Mk)$. The product $\star_{m,\mu}$ is sequentially continuous and can therefore be uniquely extended to the completion. The extended product depends smoothly on $m^2$ in the sense of (\ref{smoothness in m}).
The change of $\mu$ amounts to the transition to an equivalent product. Namely, the function
\begin{equation}
w_{m}^{\mu_1,\mu_2}=H_{m}^{\mu_1}-H_{m}^{\mu_2}
\end{equation}
is smooth. Therefore the linear isomorphism
\begin{equation}
\al_{w_{m}^{\mu_1,\mu_2}}\doteq\exp({\hbar\Ga_{w_{m}^{\mu_1,\mu_2}}})
\end{equation}
of $\Fcal_0(\Mk)$ which interpolates between the products $\star_{m,\mu_1}$ and $\star_{m,\mu_2}$ is a homeomorphism. It therefore extends to an isomorphism of $\Fcal(\Mk)$ and interpolates also the extensions of the products to this space.

In order to eliminate the dependence of the products on $\mu$ we now, as indicated in Section \ref{sec:defcon}, use the maps $\al_{H_{m}^{\mu}}$ to define an $m$-dependent, but $\mu$-independent topology on $\Fcal_0(\Mk)$ as the initial topology of these maps.
The sequential completion we denote by $\Fcal^{(m)}(\Mk)$. Elements  $F\in\Fcal^{(m)}(\Mk)$ may be identified 
(by setting $F_\mu\doteq\al_{H^{\mu}}(F)$) 
with families $(F_{\mu})_{\mu>0}$, $F_{\mu}\in\Fcal(\Mk)$, with the property
\begin{equation}
F_{\mu_1}=\al_{w_{m}^{\mu_1,\mu_2}}(F_{\mu_2}) \ .
\end{equation}

The advantage of this somewhat abstract construction is that  $m^2$ is the only scale in the algebra 
$\A^{(m)}(\Mk)=(\Fcal^{(m)}(\Mk),\star_m)$. The other possibility, namely to set $\mu^2=m^2$, would lead to singularities at $m^2=0$.     

We now define the following bundle of algebras,
\[
\Bcal = \bigsqcup_{m^2\in\R}\A^{(m)}(\Mk)\ .
\] 
Smooth sections $A=(A_m)_{m^2\in\R}$ of this bundle are, by definition, sections with the property that 
$\al_{H^\mu}(A)$, with 
\[
\al_{H^{\mu}}(A)_m=\al_{H_{m}^{\mu}}(A_m) \ ,\ m^2\in\R\ ,
\]
is a smooth function of $m^2$. Again, the property of being smooth is independent of the 
choice of $\mu$. The algebra of smooth sections is denoted by $\A(\Mk)$. $\A_0(\Mk)$ is the subalgebra of sections taking values in $\Fcal_0(\Mk)$.

\subsection{Local functionals and interactions}
Local functionals were briefly discussed before, see eq.\eqref{local}, and we want now to make the appropriate definitions and present results used later on.

A map $F:\Ecal(\Mk)\rightarrow\C$ is said to be a {\em local} functional  if it satisfies the following requirements;
\begin{enumerate}
\item $F$ satisfies the following {\em additivity} property:
\[
F(\varphi+\chi+\psi)=F(\varphi+\chi)-F(\chi)+F(\chi+\psi)\ ,
\]
if $\supp(\varphi)\cap\supp(\psi)=\emptyset$,
\item $F$ is infinitely differentiable;
\item $\WF(F^{(n)}(\varphi))\perp T\Delta_n\ $ where $\Delta_n\doteq\{(x_1,...,x_n)\in\Mk^n\,|\,x_1=...=x_n\}$.
\end{enumerate}
The space of local functionals is termed $\Fcal_{loc}(\Mk)$.

Now, the first consideration is whether additivity implies the support property stated in \eqref{local}. Indeed,
\begin{lemma}
Local functionals have the property that their $n$th order functional derivatives $F^{(n)}(\varphi)$ are supported on thin diagonals $\Delta_n$, $n\in\NN$.
\end{lemma}
\begin{proof}
By definition of functional derivative of $n$th order we have
\begin{equation}
\langle F^{(n)}(\varphi),\psi_1\otimes\cdots\otimes\psi_n\rangle=\left. \frac{d^n}{d\lambda_1\cdots d\lambda_n}\right\arrowvert_{\lambda_1=\cdots=\lambda_n=0} F\left (\varphi+\sum_{i=1}^{n}\lambda_i \psi_i\right )\ .\label{eq:definder}
\end{equation}
The support of $F^{(n)}(\varphi)$ is composed by $n$tuple of points $(x_1,\dots,x_n)\in\Mk^n$. Let us assume that in the $n$tuple one can find two points $x_j,x_k$  with $x_j\neq x_k$. Now, there exist two smooth functions $\psi_j, \psi_k$ such that $x_j\in\supp(\psi_j)$, $x_k\in\supp(\psi_k)$ and $\supp(\psi_j)\cap\supp(\psi_k)=\emptyset$. By use of the additivity condition in \eqref{eq:definder}, one sees that each term of the resulting sum would not contain all $\lambda$'s and the derivatives will all be zero. Hence the support can only be those $n$tuple of points $(x_1,\dots,x_n)$ for which $x_1=\cdots=x_n$.
\end{proof}

Another important property is the following;
\begin{lemma}\label{lemma:splittinglocfunct}
Any local functional $F$ can be written as a finite sum of local functionals of arbitrarily small supports.
\end{lemma}
\begin{proof}(Cf. \cite{DFloop} for a similar argument)
Let $\epsilon>0$. Let $(B_i)_{i=1,\ldots,n}$ be a finite covering of $\supp F$ by balls of radius $\epsilon/4$ and let $(\chi_i)_{i=1,\ldots,n}$ be a subordinate partition of unity. By a repeated use of the additivity of $F$ we arrive at a decomposition of the form
\begin{equation}
F=\sum_I s_I F_I
\end{equation} 
with $s_I\in\{\pm 1\}$, $F_I(\varphi)=F(\varphi\sum_{i\in I}\chi_i)$ and where $I$ runs over all subsets of $\{1,\ldots,n\}$ such that $B_i\cap B_j\neq\emptyset$ for all $i,j\in I$.
From the definition of the support of a functional we immediately find $\supp F_I\subset \bigcup_{i\in I}B_i:=B_I$. Since any two points in $B_I$ have distance less than $\epsilon$, each $B_I$ is contained in a ball of radius $\epsilon$. 
\end{proof}
The previous lemma will find application in the next section and in Appendix B.
\vskip.3truecm
\noindent The possible interactions $A$ for a quantum field theory build a subspace $\A_{loc}(\Mk)$ of 
$\A(\Mk)$. It is characterized by the requirement that $\al_{H^{\mu}}(A)$ is a local functional for some, and hence for all $\mu$.

Functional derivatives on $\A(\Mk)$, as well as on $\A_{loc}(\Mk)$, can be introduced as linear maps from $\Ecal(\Mk)$ to $\A(\Mk)$ by
\begin{equation}\label{eq:Aderivative}
\langle \frac{\de}{\de\ph}A,\psi\rangle=\alpha_H^{-1}\langle\frac{\de}{\de\ph}\alpha_H A,\psi\rangle \ ,
\end{equation}
since the right hand side is independent of $H$.

\section{Causal Perturbation Theory -- the Epstein-Glaser Method }
While local interactions are more singular, they also have nice properties which one can exploit for a 
perturbative construction of interacting quantum field theories \cite{EG,BF,DFret,qap}.
We first collect some properties of the $S$-matrix defined in (\ref{S-matrix}). As much as we did in eq.(\ref{eq:support}), we associate to every $A\in\A(\Mk)$ a compact region (denoted as $\supp(A)$ by abuse of notation) as the set 
\[
\supp(A)\doteq\supp(\alpha_H(A))\ .
\]
Notice that $\supp(A)$ does not depend on the choice of $H$, since the  homeomorphisms $\alpha_w$ do not change the support of a functional.

\subsection{Renormalization and the Main Theorem}
We use the fact that for $A,B\in\A_0(\Mk)$ with $\supp(A)$ later than $\supp(B)$ the time ordered product coincides with the $\star$-product
\begin{equation}\label{time ordering}
A\cdot_T B = A\star B \ .
\end{equation}
This implies the following causality property of the $S$-matrix
\begin{quote}
\item[{\bf C1. Causality}.] $S(A+B)=S(A)\star S(B)\ $ if $\supp(A)$ is later than $\supp(B)$. 
\end{quote}
The causality property determines the derivatives $S^{(n)}$ of $S$ at the origin 
\[\left. S^{(n)}(0)(B^{\otimes n})\equiv S^{(n)}(B^{\otimes n})\equiv\frac{d^n}{d\la^n}S(\la B)\right |_{\la=0} \ ,\]
(i.e.~the higher order time ordered products) partially in terms of lower order derivatives 
namely
\begin{equation}
S^{(n)}(A^{\otimes k}\otimes B^{\otimes (n-k)})= S^{(k)}(A^{\otimes k})\star S^{(n-k)}(B^{\otimes (n-k)})\ .
\end{equation}
While on $\A_0(\Mk)$ this is an immediate consequence of the definition of the $S$-matrix and of (\ref{time ordering}), 
it is the key property by which an extension to local functionals can be made,
 i.e. $S:\A_{loc}(\Mk)\rightarrow \A(\Mk)$ can be defined. 
Namely, by Lemma~\ref{lemma:splittinglocfunct}, local functionals can be splitted into a sum of terms which are localized in smaller regions. Together with the multilinearity of the higher derivatives this allows the determination of the $n$th order in terms of the derivatives with order less than $n$
for all elements of the tensor product $\A_{loc}(\Mk)^{\otimes n}$ whose support is disjoint from the thin diagonal.
Here the support of $\sum(A_1^{i}\otimes \cdots \otimes A_n^{i})$ is defined as the union of the cartesian products of the supports of $A_k^{i}$.   
Together with the property 
\begin{quote}
\item[{\bf C2. Starting element}.] $S(0)=1$, $S^{(1)}=\id\ $,
\end{quote}
this fixes the higher derivatives of $S$ at the origin partially on local functionals. 

The $\star$-product and the time ordered product $\cdot_T$ on $\Fcal_0(\Mk)$ were defined in terms of 
functional differential operators. Therefore the $S$-matrix $S(V)$, $V\in\Fcal_0(\Mk)$ at the field configuration $\ph$ depends on $\ph$ only via the functional derivatives of $V$ at $\ph$. We require that a similar condition holds true also for the extension of $S$ to $\A_{loc}(\Mk)$.

Let $V\in\A_{loc}(\Mk)$. The Taylor expansion of $\al_H(V)$ at $\ph=\ph_0$ up to order $N$ is
\begin{equation}
\al_H(V)_{\ph_0}^{(N)}(\ph)=\sum_{n=0}^{N}\frac{1}{n!}\left\langle\frac{\delta^n\al_H(V)}{\delta\ph^n}(\ph_0),(\ph-\ph_0)^{\otimes n}\right\rangle \ .
\end{equation}
We impose the following condition:
\begin{quote}
\item[{\bf C3. $\ph$-Locality}.] $\al_H\circ S(V)(\ph_0)=\al_H\circ S\circ \al_H^{-1}(\al_H(V)_{\ph_0}^{(N)})(\ph_0)+O(\hbar^{N+1})$.
\end{quote}
The condition is independent of the choice of the Hadamard function $H$. 
For polynomial functionals $V$ the condition is, up to the information on the order in $\hbar$, empty.
The main profit of {\bf C3} is that, for the computation of a certain coefficient in the $\hbar$-expansion of
$\al_H\circ S(V)$, we may replace $\al_H(V)(\ph)$ by a {\it polynomial} in $\ph$.

The $\ph$-Locality of the extension allows a rather explicit construction. 
Namely, by Lemma 3.1, for a local interaction $V$ the $n$th functional derivative of 
$\al_{H}(V)$ is, for every field configuration $\ph$, a distribution with support on the thin diagonal $\Delta_n\subset \Mk^n$. As shown in \cite{BF}, the condition on the wave front set implies that such a distribution can be restricted to transversal surfaces where the restrictions have support in a single point. Hence the $n$th functional derivative has the form
\begin{equation}\label{eq:dfrest}
\frac{\de^n \al_{H}(V)}{\de \ph^n}(x_1,\ldots,x_n)=\sum_k V^{H,n}_{k}(x)p_k(\partial_{\text{rel}})\de(x_{\text{rel}})\ ,
\end{equation}
with the center of mass $x=\frac{1}{n}\sum x_i$, finitely many test functions  $V^{H,n}_{k}(x)$ (which depend on $\ph$) and a basis $(p_k)$ of homogeneous symmetric polynomials in the derivatives with respect to relative coordinates  $x_{\text{rel}}$. 

\vskip10pt
{\small
\noindent{\bf Example.} Let $F(\ph)=\frac12\int dx f(x)\ph(x)\partial\ph(x)\partial\ph(x)$, with a test function $f\in \Dcal(\Mk)$. The second functional derivative of $F$ is a symmetric distribution in two variables, characterized by the condition
\[
\left\langle \frac{\de^2 F}{\de\ph^2}, h\otimes h\right\rangle = \left.\frac{d^2}{d\lambda^2}\right|_{\lambda=0} F(\ph+\lambda h)\ ,\quad h\in\Ccal(\Mk)\ .
\]
We compute
\begin{align*}
\left.\frac{d^2}{d\lambda^2} \right|_{\lambda=0} F(\ph+\lambda h)
&=\int dx f(x) (2\partial\ph(x)\partial h(x) h(x) 
+ \partial h(x) \partial h(x) \ph(x))\\
&=\int dx_1 dx_2 \de(x_1 -x_2) f(x)
\Bigl (\partial\ph(x)(\partial h(x_1) h(x_2)+\partial h(x_2) h(x_1))\\
&\qquad\qquad\qquad\qquad\qquad\qquad\qquad+\ph(x)(\partial h(x_1)\partial h(x_2))\Bigr)\ ,
\end{align*}
where $x=(x_1+x_2)/2$ is the center of mass coordinate.
Formal integration by parts gives
\begin{align*}
&\left.\frac{d^2}{d\lambda^2} \right|_{\lambda=0} F(\ph+\lambda h)
=\int dx_1 dx_2  h(x_1) h(x_2) \Bigl (-(\partial_{x_1} +\partial_{x_2})\left(\de(x_1 -x_2) f(x)\partial\ph(x)\right)\\
&\quad\qquad\qquad\qquad\qquad\qquad\qquad\qquad\qquad\qquad\qquad 
+ \partial_{x_1}\partial_{x_2}\left( \de(x_1-x_2)f(x)\ph(x)\right)\Bigr)\ .
\end{align*}
Introducing relative coordinates $\xi=x_1-x_2$ yields
\begin{align*}
\frac{\de^2 F}{\de\ph^2}\left(x+\frac{1}{2}\xi, x-\frac{1}{2}\xi\right)
=\left (-\partial (f\partial \ph)(x)
+ \frac{1}{4}(\partial \partial (f\ph))(x)\right )\de(\xi)
 -(f\ph)(x)(\partial \partial \de)(\xi) 
\end{align*}
which is of the form of \eqref{eq:dfrest}.}
\vskip10pt
\noindent

From (\ref{eq:dfrest}) we obtain the following expansion of $\alpha_H\circ S^{(n)}$,
\begin{equation}\label{eq:sn}
\al_H\circ S^{(n)}(V^{\otimes n})(\ph_0)=\sum_{k_j,l_j} \int \bigl(\prod_j dx_j V^{H,k_j}_{l_j}(\ph_0)(x_j)\bigr) t_l^k(\ph_0)(x_1,\ldots,x_n)
\end{equation}
where $t_l^k(\ph_0)$ is a distribution which is given by
\begin{equation}\label{eq:tn}
t_l^k(\ph_0)(x_1,\ldots,x_n)=\al_H\circ S^{(n)}(A_{l_1}^{H,k_1}(x_1)\otimes\cdots\otimes A_{l_n}^{H,k_n}(x_n)) (\ph_0)
\end{equation}
with the balanced fields \cite{BOR} (normal ordered with respect to $H$, shifted by $\ph_0$)
\begin{equation}\label{eq:bf}
\left. A_l^{H,k}(x)=\al_H^{-1}\left(p_l(-\partial_{\text{rel}})\frac{(\ph-\ph_0)(x_1)\cdots(\ph-\ph_0)(x_k)}{k!}\right) \right|_{x_1=\cdots=x_k=x}  \ .
\end{equation}

To derive (\ref{eq:sn}) we insert
\begin{equation}  
V(\ph)=\sum_{k,l} \int dx\, V^{H,k}_l(\ph_0)(x)\, A_l^{H,k}(\ph)(x) 
\end{equation}  
into $\al_H\circ S^{(n)}(V^{\otimes n})(\ph)$ and use {\bf C3} as well as linearity of $S^{(n)}$:
\begin{align} 
\al_H\circ & S^{(n)}(V^{\otimes n})(\ph)\\
&=\sum_{k_j,l_j} \int dx_1\cdots  V^{H,k_1}_{l_1}(\ph_0)(x_1)\cdots  
\al_H\circ S^{(n)}(A_{l_1}^{H,k_1}(x_1)\otimes\cdots)(\ph)\ . 
\end{align} 
Setting $\ph=\ph_0$ it results (\ref{eq:sn}).

A convenient additional condition is that, loosely speaking, $S$ should have no explicit dependence on $\ph$. Using the definition in eq.\eqref{eq:Aderivative},
\begin{quote}
\item[{\bf C4. Field Independence}.]  $\langle\de S(V)/\de\ph,\psi\rangle=S^{(1)}(V)\langle\de V/\de\ph,\psi\rangle\ ,\ $ with $V\in\A_{loc}(\Mk)$.
\end{quote} 
For the action on $\A_0(\Mk)$ this is the case due to the fact that the differential operators 
$\Ga_\cdot$ in terms of which time ordering, $\star$-product and topology were defined do not depend on $\ph$. 
In the formulae \eqref{eq:sn}, \eqref{eq:tn} the distributions $t_l^k$ become therefore independent of $\ph_0$. Hence,
in the resulting expansion of $\al_H\circ S^{(n)}(V^{\otimes n})$ the field dependence is only in the
$V^{H,k_j}_{l_j}(x_j)$. 
By the conditions C3 and C4 the construction of time ordered products is reduced to the construction 
of the  $t_l^k(\ph =0)$, i.e.~to the construction of time ordered products of balanced fields at 
$\ph=0$ , hence, the methods and results of e.g.~\cite{DFret} can be applied. 

\vskip10pt
{\small
\noindent{\bf Example.} Let $\al_H(V)=\int dx \sum_k f_k(x) \ph(x)^k/k!$, $f\in\Dcal(\Mk)$. Then
\[
\frac{\de^n \al_H(V)}{\de\ph^n}(x_1,\dots,x_n)= \sum_{k\ge n} f_k(x_1) \frac{\ph(x_1)^{k-n}}{(k-n)!}\de(x_1-x_2)\cdots\de(x_{n-1}-x_n)
\]
and now formula (\ref{eq:sn}) is determined by the expression
\[
V^{H,k}_n(x) = f_k(x) \frac{\ph(x)^{k-n}}{(k-n)!}
\]
and where the formula (\ref{eq:tn}) takes the form
\[
t^{k_1,\dots,k_n}(x_1,\dots,x_n) = \al_H\circ
S^{(n)}\circ\al^{-1}_H \left (\frac{\ph(x_1)^{k_1}}{k_1!}\otimes\cdots\otimes\frac{\ph(x_n)^{k_n}}{k_n!}\right)(\ph=0)\ . 
\]

\noindent If one replaces $H$ by $\Delta_1$, $t^H$ becomes the vacuum expectation value of the time ordered product of Wick powers, hence one obtains the Wick expansion formula of Epstein-Glaser \cite{EG}.}
\vskip10pt
The result of the Epstein-Glaser Theory is that the derivatives $S^{(n)}$ of $S$ at 0 can be extended to the full tensor 
product but that the extension is not unique. The ambiguity is described by the St\"uckelberg-Petermann Renormalization Group 
$\mathcal{R}_0$ which is the group of analytic maps of $\A_{loc}(\Mk)[[\hbar]]$ into itself with the properties
\begin{align}
Z(0)   &=0  \\
Z^{(1)}(0)  &=\id  \\
Z        &= \id + O(\hbar) \\
Z(A+B+ C)&= Z(A+B)-Z(B)+Z(B+C),\text{if $\supp(A)\cap\supp(C)=\emptyset$}\label{Zloc1}\\
\text{$\varphi$-locality}& \text{ in the} \text{ sense of C3}\\
\de Z/\de\varphi &= 0\label{Zindep}
\end{align}
The property (\ref{Zindep}) implies 
that $Z$ preserves the localization region of the interaction,
\begin{equation}
\supp(Z(V))  =  \supp(V)\ , \qquad V\in\A_{loc}(\Mk)[[\hbar]]\ ,\label{Zloc2}
\end{equation}
as may be seen from 
\begin{equation}
\langle\frac{\de}{\de\ph}Z(V),\psi\rangle=Z^{(1)}(V)\langle\frac{\de}{\de\ph}V,\psi\rangle \ .
\end{equation}

The additivity \eqref{Zloc1} expresses locality of $Z$. If one sets $B=0$ in the relation and uses $Z(0)=0$ one 
obtains the condition previously adopted in \cite{DFret}. Actually, within perturbation theory, 
the two conditions are even equivalent (for a proof see Appendix B).

In any case, since the formalism 
here adopted is different from the one in the cited reference, we recall the Main Theorem and
sketch its proof. Roughly speaking the main statement of this theorem is that in terms of the $S$-matrix a change of the renormalization prescription 
can be absorbed in a renormalization $Z$ of the interaction, where $Z$ is an element of the St\"uckelberg-Petermann Renormalization Group
$\mathcal{R}_0$. Similarly to the notations used for the derivatives of the $S$-matrix we use the shorthand notation $Z^{(n)}\equiv Z^{(n)}(0)$.

\begin{teo}[Main Theorem of Renormalization]
Given two $S$-matrices $S$ and $\widehat{S}$ 
satisfying the conditions Causality, Starting Element, $\varphi$-locality, and 
Field Independence, there exists a unique $Z\in\mathcal{R}_0$ such that 
\begin{equation}
\widehat{S}=S\circ Z \ .\label{mainthm}
\end{equation}
Conversely, given an $S$-matrix $S$ satisfying the 
mentioned conditions and a $Z\in\mathcal{R}_0$, Eq. (\ref{mainthm})
defines a new $S$-matrix $\widehat{S}$ satisfying also these conditions.
\end{teo}
\begin{proof}
Since the last part is obvious, we provide hints for the first part by following \cite{DFret}.
So, let us assume that the first $n$ elements of the formal power series for $Z$ are given, i.e. $Z^{(k)}$, $k\le n$, and define a 
second sequence 
\begin{align}\label{Z-subdiag}
Z_n^{(k)} &\doteq
\begin{cases}
Z^{(k)}\ , & k\le n\ , \\
0\ , & k > n\ .
\end{cases}
\end{align}
The corresponding $Z_n$ is an element of $\mathcal{R}_0$. Hence, by the last part of the Theorem,
\[
\widehat{S}_n \doteq\ S \circ Z_n \ . 
\]
is an admissible $S$-matrix which coincides with $S$ in lower orders $k<n$. Hence
\begin{equation}
Z^{(n+1)} \doteq \widehat{S}^{(n+1)}-\widehat{S}_n^{(n+1)}  \label{Z-induction}
\end{equation}
is an element of $\A_{loc}(\Mk)[[\hbar]]$, which is of order $\hbar$, satisfies locality (\ref{Zloc1}) and field independence 
(\ref{Zindep}). Therefore, we may use (\ref{Z-induction}) to continue the inductive construction of $Z$. 
\end{proof}

The ambiguity described in the St\"uckelberg-Petermann Renormalization Group can be reduced by imposing further 
renormalization conditions. One of the conditions is 
\begin{quote}
\item[{\bf C5. Unitarity}.] $\overline{S}(-V)\star S(V)=1\ ,$
\end{quote}
where $\overline{S}$ is the anti time ordered exponential, $\overline{S}(V)=\overline{S(\overline{V})}$, hence $S(V)$ is unitary for imaginary interactions $V$. This condition can always be fulfilled. It restricts the renormalization group to elements $Z$ which satisfy the equation $\overline{Z}(-V)+Z(V)=0$.

\subsection{Symmetries}
In general, if a symmetry $g$ acts as an automorphism of  
$\A(\Mk)$, commutes with complex conjugation and leaves the set of localized elements invariant, 
it will transform $S$ to another $S$-matrix $\widehat{S}=g\circ S\circ  g^{-1}$ satisfying also conditions C1-C5. Therefore
its effect can be described by an element $Z(g) \in\mathcal{R}_0$, $\widehat{S}=S\circ Z(g)$. 
If the symmetries form a group $G$, one obtains in this way a cocycle in $\mathcal{R}_0$,
\begin{equation}
Z(gh)=Z(g)gZ(h) g^{-1} \ .
\end{equation}
Provided the cocycle is a coboundary, i.e.~there exists an element $Z\in\mathcal{R}_0$ such that
\begin{equation}
Z(g)=Z g Z^{-1}g^{-1} \quad\quad\forall g\in G \ ,
\end{equation}
the $S$-matrix $S\circ Z$ is invariant. 

In many cases the existence of a symmetric $S$-matrix just follows from the fact that the cohomology of the group in question is trivial. This holds in particular for amenable groups, where the trivializing element $Z$
can be obtained by integrating the cocycle over the group.
 
An important case, when the group is not amenable 
is that of the Poincar\'e group  $P^\uparrow_+$. Here the result applies that the cohomology is trivial if all finite dimensional 
representations of the group are completely reducible (see e.g.~Appendix D of \cite{DFret}).
The invariance under the Poincar\'e group can be imposed as a further condition:
\begin{quote}
\item[{\bf C6. Invariance}.] $S$ is Poincar\'e invariant.
\end{quote}

A crucial example where the cocycle can (and will) be nontrivial is the group of scaling transformations 
$\R_+$. 
Scaling transformations on generic spacetimes can be encoded in a scaling of the spacetime metric (see \cite{HW3}). If one restricts the formalism to Minkowski space as we are doing it here, it is more natural to scale the points in Minkowski space after fixing some origin. This leads to the following  action on field configurations $\ph\in\Ccal(\Mk)$ \begin{equation}
(\si_\rho \ph)(x)= \rho^{\frac{2-d}{2}}\ph(\rho^{-1} x)\label{def-sigma}
\end{equation}
with the spacetime dimension $d$. 
This induces an action on $\Fcal_0(\Mk)$ which is continuous with respect to the topology 
(\ref{topology1},\ref{topology2}).  Moreover, it transforms the products $\star_m$ into $\star_{\rho m}$, and induces a linear isomorphism between the completions $\Fcal^{(m)}(\Mk)$ and $\Fcal^{(\rho m)}(\Mk)$. It therefore gives rise to an action by automorphisms of $\A(\Mk)$ defined by
\begin{equation}
\si_{\rho}(A)_m=\si_{\rho}(A_{\rho^{-1}m}) \ .
\end{equation}

On the basis of these arguments, let 
\begin{equation}
\si_{\rho}\circ S\circ \si_{\rho}^{-1}=S \circ Z(\rho) \ .\label{scalingS}
\end{equation}
Then $Z(\rho)$ satisfies the cocycle condition
\begin{equation}\label{cocycle}
Z(\rho_1\rho_2)=Z(\rho_1)\sigma_{\rho_1} Z(\rho_2)\sigma_{\rho_1}^{-1} \ .
\end{equation}
The nontriviality of this cocycle is just the well known scaling anomaly. 
One may replace the condition of scale invariance which cannot be fulfilled in general by the condition of almost scale invariance. In terms of the definition (\ref{scalingS}) above $S$ is called almost scale invariant if
\begin{quote}
\item[{\bf C7. Scaling}.] 
\[
\left(\rho\frac{d}{d\rho}\right)^n Z(\rho)=O(\hbar^{n+1}) \ ,\ 
n\in\mathbb{N} \ .
\]
\end{quote}
%
\section{The Renormalization Groups}
The formalism so far developed is flexible enough to allow a
comparison among different formulations of the idea of  the renormalization group. Similar ideas like those exposed here are already present in the literature of quantum field theory (see, e.g.~\cite{DJPS} for the cocycle case), although mainly using the euclidean formalism and viewing quantum field theory as a problem in statistical mechanics. We emphasize that in our setting the comparison can be done directly in terms of the physical spacetime, allowing a possible extension of the techniques and results to situations beyond the control of the euclidean framework.  

\subsection{The Gell-Mann-Low cocycle}
Our main intention here is to study the effect of scaling on the renormalization group, which may now be restricted to the subgroup $\mathcal{R}\subset\mathcal{R}_0$ of all $Z$ leaving in addition the conditions C5, C6 and C7 invariant. They are those elements which 
fulfil $\overline{Z}(-V)+Z(V)=0$, are Poincar\'{e} invariant and satisfy the condition
\begin{equation}
\left(\rho\frac{d}{d\rho}\right)^n\sigma_\rho\circ Z \circ\sigma_\rho^{-1}=O(\hbar^{n+1})
\end{equation}
For $m=0$, this implies, together with the facts that $Z$ maps local fields into local fields and that local fields scale homogeneously, that $Z$ is even scale invariant. Smoothness in the mass now allows it to draw a similar conclusion in the case of nonzero masses. In even spacetime dimensions $d$,
 as was shown in \cite{DFret}, if one uses the Hadamard function $H^{\mu}$, the transformed renormalization group element
\begin{equation}
Z_{H^{\mu}}=\al_{H^\mu}\circ Z\circ \al_{H^\mu}^{-1}\label{Zmu}
\end{equation}
is actually scale invariant. If we exhibit the dependence on the mass $m$ this means
\[
\sigma_\rho\circ Z_{H^{\mu}}^{(m)}\circ\sigma_\rho^{-1}=Z_{H^{\mu}}^{(\rho m)}\ ,
\]
hence the parameter $\mu$ is not scaled.
Using the transformation properties of $H^\mu$ under scaling
\[
\al_{H^{\mu}}^{-1}\circ\sigma_\rho\circ\al_{H^\mu}=\al_{H^{\mu}}^{-1}\circ\al_{H^{\rho\mu}}\circ\sigma_\rho=\al_{H^{\rho\mu}-H^{\mu}}\circ\sigma_\rho \  ,
\]
we arrive at the explicit scale dependence of $Z\in\mathcal{R}$ 
\begin{equation}
\label{scaling}
\sigma_\rho\circ Z \circ\sigma_\rho^{-1}=\al_{v\log\rho^2}^{-1}\circ Z\circ \al_{v\log\rho^2} \ ,
\end{equation}
where $v$ is the smooth function in eq.~(\ref{def(v)}).

Unfortunately, the claim in \cite{DFret}, that $Z_{H^\mu}$ is independent of $\mu$, is true 
in general only in low orders in $m^2$ (depending on the dimension). $Z$, as defined above, is of course independent of $\mu$, but 
no longer scale invariant. 

We can now analyze the cocycle $Z(\rho)$ of renormalization transformations characterising the dependence of $S$ under scale transformations.

Using (\ref{scaling}) the cocycle relation (\ref{cocycle}) takes the form 
\begin{equation}
Z(\rho\tau)=Z(\rho)\circ \al_{v\log\rho^2}^{-1}\circ Z(\tau)\circ \al_{v\log\rho^2} \ .\label{GLM-cocycle}
\end{equation}
 The cocycle $(Z(\rho))_{\rho>0}$ may be decomposed into two 1-parameter groups such that one of them 
becomes trivial in the limit $m\to 0$ and the other one converges to $Z(\rho)$, 
\begin{equation}
Z(\rho)=\hat{Z}(\rho)\circ \al_{v\log\rho^2} \ .\label{1pargroup}
\end{equation}
The one parameter group $(\hat{Z}(\rho))_{\rho>0}$ was found by Hollands and Wald \cite{HW3}. 
It is, however, not a subgroup of the St\"uckelberg-Petermann Renormalization Group as defined 
above, since the linear term of $\hat{Z}(\rho)$ is not the identity.

The $\beta$-function of standard perturbation theory is closely related to the 
generator of the one-parameter group $(\hat{Z}(\rho))_{\rho>0}$, which we call the $\hat B$-function,
\[
\left. \hat B = \rho \frac{d}{d\rho}\hat{Z}(\rho)\right |_{\rho=1}=
\left. \rho \frac{d}{d\rho}Z (\rho)\right |_{\rho=1}-2\hbar\Gamma_v\ .
\]
The $\hat B$-function is analytic, since $Z(\rho)$ has this property. In the Taylor series
the first order term is given by $\Gamma_v$ and the higher order terms by derivatives of $Z(\rho)$: 
\begin{equation}
\hat B(V)=-2\hbar\Gamma_v\, V+\sum_{n=2}^\infty\frac{1}{n!}\,\hat B^{(n)}(V^{\otimes n})\ ,\label{expansionB}
\end{equation}
where
\begin{align*}
\hat B^{(n)}(0)(V^{\otimes n})\equiv \hat B^{(n)}(V^{\otimes n})&=\left. \frac{d^n}{d\lambda^n}\right |_{\lambda=0}\,\hat B(\lambda V)\\
&=\left. \rho \frac{d}{d\rho}\right |_{\rho=1}\,
\left. \frac{d^n}{d\lambda^n}\right |_{\lambda=0}\, Z(\rho)(\lambda V)
\end{align*}      
for $n\geq 2$. 

The action of the one-parameter group $(\hat{Z}(\rho))_{\rho>0}$ on $\A(\Mk)$ 
can now be obtained as a solution of the differential equation
\[
\rho \frac{d}{d\rho}\hat{Z}(\rho)= \hat B \circ \hat{Z}(\rho)\ .
\]

Let us discuss a simple example of a $\hat B$-function:

\vskip15pt
{\small\noindent
{\bf Example} Let $V=\frac{ig}{\hbar}\int dx f(x)\ph(x)^2$ in $d=4$ dimension, with $f\in\Dcal(\Mk)$. For this interaction, 
renormalization is necessary only for the so-called fish diagram. The undetermined term does not depend on $m$. 
Hence the computation of $Z(\rho)(V)$ can be performed at $m=0$, see (\ref{B-fisch}). It results
\begin{eqnarray*}
Z(\rho)(V) & = & V + ig^2\,\frac{\log\rho}{8\pi^2} \int dx f(x)^2\ , \\
\al_{-v\log\rho^2}(V) & = & V -ig\, \frac{m^2}{4(2\pi)^2} \,\log\rho^2 \int dx f(x)\ ,
\end{eqnarray*}
where we used $v(x,x)=m^2/4(2\pi)^2$ (\eqref{def(v)} and Appendix A), hence
\[
\hat B(V) = i \int dx \left( - g\, \frac{m^2}{2(2\pi)^2}\, f(x) + g^2\,\frac{f(x)^2}{8\pi^2}\right)\ .
\]
}
\vskip10pt
The Gell-Mann-Low cocycle $(Z(\rho))_{\rho>0}$ and the corresponding one parameter group $(\hat{Z}(\rho))_{\rho>0}$
depend on the chosen renormalization prescription $S$. The $\hat B$-functions belonging to different 
renormalization prescriptions are related as follows.
\begin{lemma}
Let $S_1$ and $S_2$ be two $S$-matrices and let $Z\in\mathcal{R}$ be the corresponding renormalization 
group transformation: $S_2=S_1\circ Z$. Let $(\hat{Z}_1(\rho))_{\rho>0}$ and $(\hat{Z}_2(\rho))_{\rho>0}$
be the pertinent one parameter groups (\ref{1pargroup}). Their generators $\hat B_1$ and $\hat B_2$, respectively, are related by
\begin{equation}
Z\circ \hat B_2= \hat B_1\circ Z\ .\label{relationB}
\end{equation}
To lowest non-trivial order this relation reads
\begin{equation}
\hat B_1^{(2)}-\hat B_2^{(2)}=4\hbar^2\,Z^{(2)}\circ(\Gamma_v\otimes\Gamma_v)+2\hbar\,\Gamma_v\circ Z^{(2)}\ .\label{relationB^2}
\end{equation}
In the massless case the r.h.s.~vanishes, i.e. $\hat B^{(2)}$ is universal.
\end{lemma}
\begin{proof}
Applying a scale transformation to  $S_2=S_1\circ Z$ and using (\ref{scalingS}) and (\ref{scaling}) we obtain
$S_2\circ\hat Z_2(\rho)=S_1\circ \hat Z_1(\rho)\circ Z\,$, from which we conclude
\[
Z\circ\hat Z_2(\rho)=\hat Z_1(\rho)\circ Z\ .
\]
Application of $\rho \frac{d}{d\rho}|_{\rho=1}$ yields immediately the assertion (\ref{relationB}). Inserting 
the lowest order terms of the Taylor expansion of $Z$ and $\hat B_1,\,\hat B_2$ (\ref{expansionB}), respectively, 
it results (\ref{relationB^2}). 
\end{proof}

In odd spacetime dimensions $d$, the Hadamard function $H^{(d)}$ scales homogeneously.
Therefore, all $Z\in\mathcal{R}$ are scale invariant and $(Z(\rho))_{\rho >0}$ is a group 
(i.e.~$\al_{v\log\rho^2}$ does not appear in (\ref{scaling}) and (\ref{GLM-cocycle})). It 
follows that the r.h.s.~of (\ref{relationB^2}) vanishes, that is $\hat B^{(2)}$ is universal also  
for non-vanishing mass.

\subsection{Flow Equation}\label{sec:flow}

In this section we formulate the renormalization method of the flow equation in our formalism
and relate the renormalization group in the sense of Wilson to the St\"uckelberg-Petermann group. 

As we pointed out in the introduction, the time ordering prescription $T$ can be formally understood as the operator of convolution with 
the oscillating Gaussian measure with covariance $i\hbar\Delta_D$. The crucial point for us now is that if we split the covariance in two, or more pieces, say $i\hbar\Delta_D= C_1 + C_2$, then we get a semigroup law from the convolution \eqref{path integral}, namely
\begin{equation}\label{semigroup}
TF(\varphi)= \int d\mu_{C_1}(\phi_1) \left (\int d\mu_{C_2}(\phi_2) F(\varphi -\phi_2 -\phi_1)\right)\ .
\end{equation}
Hence, the idea is to split the covariance in such a way as to get more and more regular convolutions. 
The splitting is usually parametrized by a cutoff scale $\Lambda$. Since the left hand side of (\ref{semigroup}) 
is independent of $\Lambda$, the derivative of the right hand side has to vanish.
This leads to a differential equation, the Flow Equation, that was first used for the purposes of perturbative renormalization by Polchinski \cite{Polchinski}(see also \cite{Gall}). 

In our setting the procedure can be described as follows. Since the time ordered product \eqref{top} is a deformation 
of the pointwise product induced by the time ordering operator $T$, a regularization of the latter induces 
a regularized time ordered product. Hence, by a regularizing procedure, 
we interpolate between the pointwise and the time ordered products. Namely, 
let us pose $T_\Lambda \doteq \exp(i\hbar \Gamma_\La)\equiv \alpha_{h_\Lambda -H} $, 
where $(h_\Lambda)_{\Lambda\in\R_+}$ is a family of symmetric smooth functions in $\Mk^2$, which depend in a differentiable manner on the parameter 
$\Lambda$, and is such that $h_0=0$ and $h_\Lambda \to H_F$ in the sense of H\"ormander as $\Lambda\to \infty$. (We recall that $H_F\doteq i\Delta_D +H$, 
hence $h_\Lambda-H \to i\De_D$.) This means that, as initial condition we have $\alpha_H T_0 = \id$ and that in the limit $\Lambda \to \infty$ we have 
$\alpha_H T_\La\to\alpha_{H_F}$, by the sequential continuity of the maps.

By means of this family of regularized time ordering operators, we can construct a family of regularized $S$-matrices as $S_\Lambda\doteq \exp_{\cdot_{T_\Lambda}}$, as well as the family of inverses $S_\Lambda^{-1} = T_\Lambda \circ \log \circ\ T_\Lambda^{-1}$.
The principal aim of the Flow Equations is to study the behaviour of effective potentials under infinitesimal cut-off variations. 
Effective potentials $V_{\Lambda}$ at scale $\Lambda$ are thought of as arising from integrating out the degrees of freedom above $\Lambda$. In our formalism they are defined by 
$S(V)=S_\Lambda(V_\Lambda)$, i.e. $V_\Lambda = S_\Lambda^{-1}\circ S(V)$.

Now, we can prove, in our setting, that the effective potentials defined above fulfil the Flow Equation.
\begin{prop}
Let $V_\La$ be the effective potential at scale $\La$ for any local interaction $V\in\A_{loc}(\Mk)$, then
\[
\frac{d}{d\La} V_\La = -\frac{1}{2}\left(\frac{d}{d\La} \mathcal{M}_{T_\La}\right)(V_\La\otimes V_\La)\ .
\]
\end{prop}
\begin{proof}
We have, by definition
\[
V_\La = S_\La^{-1} \circ S(V) = T_\La \log T_\La^{-1} T \exp T^{-1}(V)\ ,
\]
from which one also gets that $T^{-1}_\La T\exp T^{-1}(V) = \exp T^{-1}_\La(V_\La)$. Using the previous relations 
we have the following chain of identities s
\begin{align*}
\frac{d}{d\La} V_\La &= \dot{T}_\La \log T_\La^{-1} T \exp T^{-1}(V)\\ &\qquad\qquad+ T_\La (T_\La^{-1} T \exp T^{-1}(V))^{-1} (-T_\La^{-2}) \dot{T}_\La T \exp T^{-1}(V)\\
&=i\hbar\,\bigl( \dot{\Gamma}_\La V_\La + T_\La \exp (-T_\La^{-1}(V_\La)) (-\dot{\Gamma}_\La) \exp T^{-1}_\La(V_\La)\bigr)\\
&= i\hbar\,\bigl(\dot{\Gamma}_\La V_\La + T_\La (-\dot{\Gamma}_\La) (T_\La^{-1}(V_\La))\bigr)
-\frac{1}{2}\left(\frac{d}{d\La}\mathcal{M}_{T_\La}\right)(V_\La \otimes V_\La)
\end{align*}
where in the last relation we made use of the fact that $\dot{\Gamma}_\La$ contains functional derivatives of second order and that the Hadamard function is symmetric.
Since the linear operators $T_{\Lambda}$ and $\dot{\Gamma}_{\La}$ commute,
the first two terms add up to zero and we obtain the Flow Equation.
\end{proof}

The attractive feature of the Flow Equation is that it can be immediately integrated in perturbation theory, since the $\La$-derivative of the $n$th order term in $V$ of $V_{\La}$ is determined by the terms of order less than $n$ due to $V_{\La}^{(0)}=0$. In general, however, $V_{\La}$ will not converge for $\La\to\infty$. Here we can use the insight of Epstein-Glaser that $S$ always exists on local functionals, but is not unique. Moreover, one can show, that for each $\La$ there is an element $Z_{\Lambda}$ of the
St\"uckelberg-Petermann Renormalization Group\footnote{Generically $S_\Lambda$ does not satisfy Poincar\'e
invariance C6 and the Scaling property C7, hence one may not expect that $Z_\Lambda$ is in $\mathcal{R}$.}
$\mathcal{R}_0$, such that
\begin{equation}
\lim_{\La\to\infty} S_{\La}\circ Z_{\La}=S \ .\label{counterterms}
\end{equation}
$Z_\La$ adds the local counter terms which are needed for the existence of the limit. 
One may determine $Z_{\La}$ directly from the knowledge of $S_{\La}$, namely the $n$th order term of
$S_{\La}\circ Z_{\La}$ is, due to $S_{\Lambda}^{(1)}=\id$
\begin{equation}
(S_{\La}\circ Z_{\La})^{(n)}=Z^{(n)}_{\La}+X^{(n)}_{\La}
\end{equation}
where $X^{(n)}_{\La}$ depends only on the terms of $Z_{\La}$ of order less than $n$. 
Hence  by induction $Z^{(n)}_{\La}$ is uniquely determined up to the addition of a local map 
$\de Z^{(n)}_{\La}$ which converges as $\La\to\infty$. This amounts to the freedom of finite renormalization, 
encoded in the St\"uckelberg-Petermann Renormalization Group.

From the definition of $V_\Lambda$ it immediately follows that the flow of the effective potential from $\Lambda_0$
to $\Lambda$ is given by $S_\Lambda^{-1}\circ S_{\Lambda_0}$. Heuristically, this operator is approximately equal to
$Z_\Lambda\circ Z_{\Lambda_0}^{-1}\in\mathcal{R}_0$ for $\Lambda,\Lambda_0$ big enough, due to (\ref{counterterms}).
In this sense Wilson's flow of effective potentials can be approximated by a 2-parametric subfamily of the 
St\"uckelberg-Petermann group.

As an application let us see now how one can construct the Gell-Mann-Low cocycle $Z(\rho)$ from the counter terms $Z_\Lambda$.
We assume that the regularized $S$-matrices satisfy the scaling relation
\begin{equation}
\sigma_\rho\circ S_\Lambda\circ\sigma_\rho^{-1}=S_{\rho\Lambda}\ .
\end{equation}
Since a scale transformation commutes with the limit $\La\to\infty$ in (\ref{counterterms}), it follows
\[
\lim_{\La\to\infty} S_{\rho\Lambda}\circ Z_{\rho\Lambda}\circ Z_{\rho\Lambda}^{-1}\circ
\sigma_\rho\circ Z_\Lambda\circ\sigma_\rho^{-1}=S\circ Z(\rho)\ .
\]
Taking into account that $\lim_{\La\to\infty} S_{\rho\Lambda}\circ Z_{\rho\Lambda}=S$, we conclude
\begin{equation}
\lim_{\La\to\infty}Z_{\rho\Lambda}^{-1}\circ
\sigma_\rho\circ Z_\Lambda\circ\sigma_\rho^{-1}=Z(\rho)\ .\label{ct-GML}
\end{equation}
In this way the Gell-Mann-Low cocycle $Z(\rho)$ can be obtained from the counter terms $Z_\Lambda$ 
needed to cancel the divergences of the regularized $S$-matrix. If $S$ is almost scale invariant C7, we see from
(\ref{ct-GML}) that $\sigma_\rho\circ Z_\Lambda\circ\sigma_\rho^{-1}$ differs from $Z_{\rho\Lambda}$
only by $\mathrm{log}\,\rho$-terms and terms vanishing for $\Lambda\to\infty$.


\section{The renormalization group in the algebraic adiabatic limit}
\subsection{Introduction}
In the previous sections we showed how the renormalization group arises in perturbative algebraic quantum field theory as a group of formal 
diffeomorphisms on the space  of local functionals of smooth field configurations. We want to analyze the structure in 
the so-called algebraic adiabatic limit where the interaction is induced by a Lagrangian with no explicit dependence on spacetime. 
Contrary to the adiabatic limit in the sense of operators on Fock space (``strong adiabatic limit'') or in the sense of Wightman functions 
(i.e. vacum expectation values of interacting fields) (``weak adiabatic limit'') -  both appear in the work of Epstein and Glaser -
the algebraic adiabatic limit does not suffer from any infrared problems and is in particular well defined on generic globally hyperbolic 
spacetimes \cite{BF,DFsiena,HW3}. Traditionally, in causal perturbation theory, Lagrangians are integrated against test functions in order to obtain well defined local functionals on field configurations. But the action of the renormalization group on the space of local functionals is nonlinear, hence we prefer to admit also nonlinear dependence on the test function. Lagrangians in this generalized sense can be identified with a certain class of nonlinear functionals on the test function space with values in the space of local functionals of field configurations. The space of Lagrangians in this sense is invariant under the renormalization group; moreover, Lagrangians which induce equivalent theories do so also after acting on them with the renormalization group. We work out the analogue of the $\beta$-function, compare the results obtained in our framework in a few examples with results from the literature and find agreement.
\subsection{Generalized Lagrangian}
The defining properties of a generalized Lagrangian are (partially) motivated by the corresponding properties of the 
renormalization maps $Z\in\mathcal{R}$.
\begin{df}\label{df:af}
A generalized Lagrangian $\mathscr{L}$ is a map
\[
\mathscr{L} : \Dcal(\Mk)\longrightarrow \A_{loc}(\Mk)
\]
with the following properties
\begin{itemize}
\item[AF1.] $\supp(\mathscr{L}(f))\subset\supp(f)\ $;
\item[AF2.] $\mathscr{L}(0)=0$;
\item[AF3.] $\mathscr{L}(f+g+h) =\mathscr{L}(f+g) -\mathscr{L}(g)+ \mathscr{L}(g+h)\ ,$ if $\ \supp(f)\cap\supp(h)=\emptyset\ $;
\item[AF4.] $\al_L\circ \mathscr{L}= \mathscr{L}\circ L_{\ast}$ for all  elements $L$ of the Poincar\'e group $P^\uparrow_+$. 
\end{itemize}
\end{df}
\begin{prop}
The space of generalized Lagrangians is invariant under the action
\[
\mathscr{L}\mapsto Z\circ \mathscr{L}
\]
of the renormalization group $\mathcal{R}$.
\end{prop}
\begin{proof}
If $\mathscr{L}$ is a generalized Lagrangian and $Z$ an element of the renormalization group, then $Z\circ\mathscr{L}$ obviously satisfies the 
conditions AF1,2,4. To prove that also condition AF3 is satisfied we first show the following remarkable property of the map $\mathscr{L}$
(which is a weak substitute for linearity):
\begin{equation}
\supp\, \mathscr{L}_g(f)\subset\supp\, f \label{suppL} 
\end{equation}
where $\mathscr{L}_g(f)=\mathscr{L}(f+g)-\mathscr{L}(g)$.
Namely, let $\Mk\ni x\not\in\supp\, f$. Choose $h\in\Dcal(\Mk)$ with $h=g$ in a neighbourhood of $x$ such that $\supp\, h\cap\supp\, f=\emptyset$. 
Then, from condition AF3, 
\[
\mathscr{L}(f+g)=\mathscr{L}(f+(g-h)+h)=\mathscr{L}(f+(g-h))-\mathscr{L}(g-h)+\mathscr{L}(g) \ ,
\]
hence $\mathscr{L}_g(f)=\mathscr{L}_{g-h}(f)$, 
thus $\supp\, \mathscr{L}_g(f)=\supp\, \mathscr{L}_{g-h}(f)\subset (\supp (f+g-h)\cup\supp (g-h))\not\ni x$ by using AF1.

We now use the additivity of $Z$ and find for $f$ and $h$ with disjoint supports
\begin{align*}
Z(\mathscr{L}(f+g+h))&=Z(\mathscr{L}_g(f)+\mathscr{L}(g)+\mathscr{L}_g(h))\\
 &= Z(\mathscr{L}(f+g))-Z(\mathscr{L}(g))+Z(\mathscr{L}(g+h)) \ .
\end{align*}
\end{proof}
\begin{df}\label{def:Lag}
Two generalized Lagrangians are said to induce the same interaction, $\mathscr{L}\sim\mathscr{L}'$, if 
\[
\supp(\mathscr{L}-\mathscr{L}')(f)\subset \supp\, df\ ,\quad  \forall f\in\Dcal(\Mk)\ , 
\]
(i.e. the corresponding field equations differ only by boundary terms).
\end{df}
\begin{prop}\label{RG-equivalence}
If $\mathscr{L}$ and $\mathscr{L}'$ induce the same interaction then so do $Z\circ \mathscr{L}$ and $Z\circ\mathscr{L}'$ for all renormalization group elements $Z$.
\end{prop}
\begin{proof}
Let $\mathscr{L}\sim\mathscr{L}'$ and $f\in\Dcal(\Mk)$. Let
$h\in\Dcal(\Mk)$ with $\supp\, h\cap \supp\, df=\emptyset$. We have to show
that
\[
\langle
\frac{\de}{\de\ph}(Z\circ\mathscr{L}-Z\circ\mathscr{L}')(f),h\rangle=0 \ .
\]
Let $\mathscr{L}_{\lambda}(f)=(\la \mathscr{L}+(1-\la)\mathscr{L}')(f)$.
Using the fundamental theorem of calculus, we find for the left side
\begin{align*}
\int_0^1 d\la \frac{d}{d\la}\langle
\frac{\de}{\de\ph} & Z\circ\mathscr{L}_{\lambda}(f),h\rangle\\
&=\int_0^1 d\la \langle
\frac{\de}{\de\ph}Z^{(1)}(\mathscr{L}_{\lambda}(f))(\mathscr{L}-\mathscr{L}')(f),h\rangle
\\
&=\int_0^1 d\la
\,\Bigl(Z^{(2)}(\mathscr{L}_{\lambda}(f))(\langle\frac{\de}{\de\ph}\mathscr{L}_{\lambda}(f),h\rangle
\otimes (\mathscr{L}-\mathscr{L}')(f))\\
&\qquad\qquad\qquad\phantom{\int_0^1 d }+ Z^{(1)}(\mathscr{L}_{\lambda}(f))\langle
\frac{\de}{\de\ph}(\mathscr{L}-\mathscr{L}')(f),h\rangle\Bigr) \ .
\end{align*}

Both terms in the integrand vanish because of the assumption on the
supports of $df$ and $h$. The first, since
$Z^{(2)}(\mathscr{L}_{\lambda}(f))$, as a bilinear map on
$\A_{loc}(\Mk)\times \A_{loc}(\Mk)$, vanishes if the arguments have
disjoint supports. This is the case since the support of the first factor
is contained in $\supp\, h$ and the support of the second factor is
contained in $\supp\, df$ due to the equivalence of $\mathscr{L}$ and
$\mathscr{L'}$. The second term in the integrand vanishes, since, for the
same reasons, the argument of the linear map
$Z^{(1)}(\mathscr{L}_{\lambda}(f))$ is zero.  This proves the proposition.
\end{proof}

\subsection{Observables}
We now want to investigate the action of the renormalization group on observables. Given two local functionals $V$ and $F$, the relative S-matrix
\[
S_V(F)=S(V)^{-1}\star S(V+F)
\]
is the generating functional for the time ordered powers of the (retarded) observable corresponding to $F$ under the interaction $V$. If $S$ is replaced by $\hat{S}=S\circ Z$
with a renormalization group element $Z$, we obtain
\[
\hat{S}_V(F)=S_{Z(V)}(Z_V(F))
\]
where $Z_V(F)=Z(V+F)-Z(V)$.

We first observe that 
\begin{equation}
Z_V(F)=Z_{V'}(F)\quad\mathrm{if}\quad\supp(V-V')\cap\supp\, F=\emptyset\ .\label{relationZ}
\end{equation}
This follows from the additivity property of $Z$:
\[
Z(V+F)=Z((V-V')+V'+F)=Z(V)-Z(V')+Z(V'+F) \ .
\] 
Similarly to the proof of (\ref{suppL}), the relation (\ref{relationZ}) implies 
\begin{equation} 
\supp\, Z_V(F)\subset\supp\, F \ .\label{suppZ} 
\end{equation} 

Let $[\mathscr{L}]$ denote the equivalence class of the generalized Lagrangian $\mathscr{L}$, in the sense of Def.~\ref{def:Lag}. Then we set
\[
Z_{[\mathscr{L}]}(F)=Z_{\mathscr{L}(f)}(F)
\] 
with $f\equiv 1$ on a neighbourhood of $\supp\, F$. By the remark above and by (\ref{suppL}), 
the right hand side does neither depend on $f$ nor on the choice of the Lagrangian in its equivalence class.

Let $\dc$ be a relatively compact open subregion of Minkowski space. 
The rough idea of the algebraic adiabatic limit is to achieve independence  
of the observables on the behaviour of the interaction  
outside of $\dc$ by admitting all interactions which yield the same field equation in $\dc$.  
For this purpose we define
\begin{gather}
\mathscr{V}_{\mathscr{[L]}}(\dc) \doteq \{ V\in\A_{loc}(\Mk)\ |\ \supp(V-\mathscr{L}_1(f))\cap\overline{\dc}=\emptyset ,\notag\\
\quad\quad\quad\text{ if } \mathscr{L}_1\in\mathscr{[L]}\text{ and } f\equiv 1 \text{ on }\dc \}\ ,
\end{gather}
where $\mathscr{L}$ is a generalized Lagrangian.
Note that if $\supp(V-\mathscr{L}_1(f))\cap\overline{\dc}=\emptyset$ for some $\mathscr{L}_1\in\mathscr{[L]}$
and some $f\in\Dcal(\Mk)$ with $f\equiv 1$ on $\dc$, then this holds for {\it all} $\mathscr{L}_1$ and {\it all} $f$ with these properties.
Note also that $Z_V(F)=Z_{[\mathscr{L}]}(F)$ if $V\in \mathscr{V}_{\mathscr{[L]}}(\dc)$ and $\supp\, F\subset \dc$.
In addition we point out that
\[
V\in\mathscr{V}_{\mathscr{[L]}}(\dc) \Leftrightarrow Z(V)\in\mathscr{V}_{[Z\circ\mathscr{L}]}(\dc)\ , 
\]
which follows from (\ref{suppZ}) and Proposition \ref{RG-equivalence}.

The relative $S$-matrix in the algebraic adiabatic limit is defined by
\[
S^\dc_{[\mathscr{L}]}(F)=(S_V(F))_{V\in\mathscr{V}_{[\mathscr{L}]}(\dc)}
\] 
for $F$ with $\supp\, F\subset \dc$. The interpretation as algebraic adiabatic limit relies on the following argument (see \cite{BF,DFsiena}):
$S^\dc_{[\mathscr{L}]}(F)$ is a 'covariantly constant section'  in the sense that for any
$V_1,V_2\in\mathscr{V}_{\mathscr{[L]}}(\dc)$ there exists an automorphism
$\beta$ of $\A(\Mk)$ such that
\[
\beta(S_{V_1}(F))=S_{V_2}(F)\quad\ \forall F\in\A_{loc}(\Mk)\ ,\ \supp\, F\subset \dc\ .
\]
Hence, the structure of the algebra generated by $S_V(F)\ ,\ \supp\, F\subset \dc$ is independent of
the choice of $V\in \mathscr{V}_{[\mathscr{L}]}(\dc)$.

The local algebra $\A_{[\mathscr{L}]}(\dc)$ of observables in the algebraic adiabatic limit 
is generated by the elements $S^\dc_{[\mathscr{L}]}(F)$, $\supp\, F\subset\dc$.
For $\dc_1\subset\dc_2$ the embedding $\iota_{\dc_2\dc_1}:\A_{[\mathscr{L}]}(\dc_1)\hookrightarrow \A_{[\mathscr{L}]}(\dc_2)$
is induced by $\iota_{\dc_2\dc_1}\,(S^{\dc_1}_{[\mathscr{L}]}(F))=S^{\dc_2}_{[\mathscr{L}]}(F)$.

We may now determine the action of the renormalization group on observables in the adiabatic limit. Let again $Z$ be 
an element of the renormalization group and $\hat{S}=S\circ Z$. Then
\[
\hat{S}^\dc_{[\mathscr{L}]}(F)=(\hat{S}_V(F))_{V\in\mathscr{V}_{[\mathscr{L}]}(\dc)}=
(S_{Z(V)}(Z_V(F)))_{V\in\mathscr{V}_{[\mathscr{L}]}(\dc)}
=S^\dc_{[Z\circ\mathscr{L}]}(Z_{[\mathscr{L}]}(F)) \ .
\]
We conclude as a slight generalization of a Theorem in \cite{DFret,DFbros} (cf. \cite{HW3})
\begin{teo}[Algebraic Renormalization Group Equation]
Let $\hat{\A}_{[\mathscr{L}]}(\dc)$ and $\A_{[\mathscr{L}]}(\dc)$ 
denote the algebra of observables obtained by using $\hat{S}$ and $S$, respectively. 
The pertinent renormalization group element $Z\in\mathcal{R}$ induces an isomorphism 
$\al_Z=(\al_Z^\dc)$ of the nets,
\begin{equation} \label{RGE}
\al_Z^\dc: \hat{\A}_{[\mathscr{L}]}(\dc)\to \A_{[Z\circ \mathscr{L}]}(\dc)\ , 
\end{equation}
such that $\iota_{\dc_2\dc_1}\circ\al_Z^{\dc_1}=\al_Z^{\dc_2}\circ\hat\iota_{\dc_2\dc_1}$ for $\dc_1\subset\dc_2$.
The isomorphism is given by 
\begin{equation}
\al^\dc_Z(\hat{S}^\dc_{[\mathscr{L}]}(F))=S^\dc_{[Z\circ\mathscr{L}]}(Z_{[\mathscr{L}]}(F))\ .
\end{equation} 
In particular, if $\mathscr{L}$ and $Z\circ\mathscr{L}$ induce the same interaction, $\alpha_Z$ is an automorphism.
\end{teo}
\begin{oss}[{\bf Generalized fields}]
Motivated by the nice properties of a generalized Lagrangian, we generalize the
concept of fields by admitting nonlinear dependence on the test function.
\begin{df}
A generalized field is a map $\Phi$ from the space $\Ga_0(\Tcal^\infty(\Mk))$ of smooth sections with compact support on the 
tensor bundle over $\Mk$ into the space of local functionals $\A_{loc}(\Mk)$ with the following properties
\begin{enumerate}
\item $\supp\, \Phi(f)\subset \supp\, f$;
\item $\Phi(0)=0$ ;
\item $\Phi(f+g+h)=\Phi(f+g)-\Phi(g)+\Phi(g+h)$ if $\supp\, f \cap \supp\, h=\emptyset$;
\item $\al_L\circ \Phi=\Phi\circ L_{\ast}$ for all $L\in P^\uparrow_+$.
\end{enumerate}
\end{df}
Obviously, the support property (\ref{suppL}) holds true also for generalized fields. 

\noindent It is now easy to see that any $Z\in\mathcal{R}$ and any generalized Lagrangian $\mathscr{L}$ induce via
\begin{equation}
\Phi\mapsto Z_{[\mathscr{L}]}\circ\Phi
\end{equation}
a map on the space of fields (generalized field strength renormalization) which 
satisfies, due to the algebraic renormalization group equation, the relation
\begin{equation}
\alpha_Z\circ \hat{S}_{[\mathscr{L}]}\circ\Phi=S_{[Z\circ\mathscr{L}]}\circ Z_{[\mathscr{L}]}\circ\Phi \ .
\end{equation}
\end{oss}
\subsection{Scaling}
One of the deepest insights in quantum field theory is that the invariance of the physical content of the theory under a change of the renormalization parameters (the renormalization group equation) provides information on the behaviour of the theory at different scales (Callan-Symanzik equation). Usually this relation is expressed as a differential equation for the vacuum expectation values of time ordered products of fields. In our algebraic framework we obtain a corresponding result without the necessity to incorporate information on the existence and uniqueness of vacuum states. 

In  classical field theory, scaling 
would result in replacing the Lagrangian $\mathscr{L}$ by the scaled Lagrangian
$\mathscr{L}^{\rho}$, $\mathscr{L}^{\rho}(f) =\sigma_{\rho}(\mathscr{L}(f_{\rho}))$, where $f_{\rho}(x)=f(\rho^{-1} x)$ 
denotes the scaled test 
function.\footnote{Note that $\mathscr{L}^{\rho}$ differs from $\mathscr{L}$ only by a scaling of the parameters:
$(\mathscr{L}^{(\Lambda)})^{\rho}=\mathscr{L}^{(\rho\Lambda)}$, where $\Lambda$ denotes the parameters of $\mathscr{L}$ which are assumed 
to have the dimension of a mass.} In quantum field theory one obtains instead

\begin{teo}[Algebraic Callan-Symanzik Equation]
The scaled Haag-Kastler net $\dc\mapsto\A_{[\mathscr{L}]}(\rho^{-1}\dc)$ is equivalent to the net 
$\dc\mapsto\A_{[Z(\rho)\circ\mathscr{L}^{\rho}]}(\dc)$  where the isomorphism is 
induced by $\al_{Z(\rho)}\circ\sigma_{\rho}$:
$$\al_{Z(\rho)}\circ\sigma_{\rho}(S_{[\mathscr{L}]}(\sigma_{\rho}^{-1}F))=
S_{[Z(\rho)\circ\mathscr{L}^{\rho}]}(Z(\rho)_{[\mathscr{L}^{\rho}]}(F))\ , \ F\in \A_{loc}(\dc) \ .$$
\end{teo}
\begin{proof}
We have
$$\sigma_{\rho}\circ S_{[\mathscr{L}]}\circ\sigma_{\rho}^{-1}=(S\circ Z(\rho))_{[\mathscr{L}^{\rho}]} \ . $$
Hence the net $\dc\mapsto\sigma_{\rho}\A_{[\mathscr{L}]}(\rho^{-1}(\dc))$ coincides with the net $\hat{\A}_{[\mathscr{L}^{\rho}]}$ where $\hat{\A}$ is constructed from $\hat{S}=S\circ Z(\rho)$. By the algebraic renormalization group equation (\ref{RGE}) $\al_{Z(\rho)}$ induces an equivalence between the nets $\hat{\A}_{[\mathscr{L}^{\rho}]}$ and $\A_{[Z(\rho)\circ\mathscr{L}^{\rho}]}$. Hence $\al_{Z(\rho)}\circ\sigma_{\rho}$ 
induces the claimed equivalence between the nets  $\dc\mapsto\A_{[\mathscr{L}]}(\rho^{-1}\dc)$ and $\A_{[Z(\rho)\circ\mathscr{L}^{\rho}]}$.
\end{proof}

Finally we are going to perform the wave function and mass renormalization and to define the $\beta$-function. 
For computations it is useful to represent the abstract quantities $A\in\mathcal{A}$ by the functionals 
$\alpha_{H}(A)\in\mathcal{F}$. In even dimensions $H$ depends on a parameter $\mu$. This induces an additional scale into the formalism. 
We assume that the generalized Lagrangian $\mathscr{L}$ is of the form
\[
\mathscr{L}(f)\equiv\mathscr{L}_\mu(f)=\alpha_{H^\mu}^{-1}(F(f))\ \quad (\forall f\in \Dcal(\Mk))
\]
for some $F: \Dcal(\Mk)\rightarrow \mathcal{F}_{loc}(\Mk)$ which satisfies the properties {\it AF1}-{\it AF4} in Definition 6.1.
Note that $[(\ph^2)_\mu]\equiv [\alpha_{H^\mu}^{-1}\circ\varphi^2]=[\ph^2]$ and $[((\partial\varphi)^2)_\mu]=[(\partial\varphi)^2]$
since $[c]=0$ for $c\in\C[[\hbar]]$.

The transformed Lagrangian $[Z(\rho)\circ (\mathscr{L}_\mu)^\rho]$ will in general contain a multiple $\gamma_\mu(\rho)$ 
of the kinetic term $\frac{i}{2\hbar}\,[(\partial\ph)^2]$ and a multiple $\lambda_\mu(\rho)$ of the mass term 
$\frac{i}{2\hbar}\,(\rho m)^2\,[\ph^2]$. These terms may be absorbed in the free Lagrangian
by replacing $\ph$ by 
\[\ph'\doteq (1-\gamma_\mu(\rho))^{\frac12}\ph\quad\text{(``wave function renormalization'')}\]
and $m$ by 
\[m'\doteq m\,\sqrt{\frac{1+\lambda_\mu(\rho)}{1-\gamma_\mu(\rho)}}\quad\text{(``mass renormalization'')}\ .\] 
Since the new free Lagrangian differs from the old one by $\frac{-i}{2\hbar}(\gamma_\mu
[(\partial\ph)^2]+\lambda_\mu (\rho m)^2[\ph^2])$ we have to add this term to the interaction part $[Z(\rho)\circ(\mathscr{L}_\mu)^{\rho}(\ph)]$.
The new interaction Lagrangian $[\mathscr{L}'_\mu]$ is then 
\begin{eqnarray}
[\mathscr{L}'_\mu(\ph')]& = &[Z(\rho)\circ(\mathscr{L}_\mu)^{\rho}((1-\gamma_\mu(\rho))^{-\frac12}\ph')]\notag\\
&  &\qquad -\frac{i}{2\hbar}\Bigl(\frac{ \gamma_\mu(\rho)}{1-\gamma_\mu(\rho)}\,[(\partial\ph')^2]+
\frac{\lambda_\mu(\rho)\,\,(\rho m')^2}{1+\lambda_\mu(\rho)}\,[(\ph')^2]\Bigr)\ .
\end{eqnarray}
The derivative with respect to $\log\rho$ at $\rho=1$ (with keeping $\ph'$ fixed) now defines the $\beta$-function.
Using  
\begin{equation}
B\doteq\rho\frac{d}{d\rho}\vert_{\rho=1}Z(\rho) \ ,\label{def(B)}
\end{equation} 
$\gamma_\mu(1)=0$, $\lambda_\mu(1)=0$ and $(\mathscr{L}_\mu)^\rho=(\mathscr{L}^\rho)_{\rho\mu}$ the action of $\beta$ on interaction classes reads
\begin{align}
\beta([\mathscr{L}_\mu])=[B\circ\mathscr{L}_\mu]+& \frac{\dot{\gamma}_\mu}{2}\,\langle \mathscr{L}_\mu^{(1)},\ph\rangle\notag\\ 
&\qquad -
\frac{i}{2\hbar}\Bigl(\dot{\gamma}_\mu\,[(\partial\ph)^2]+\dot{\lambda}_\mu\, m^2\,[\ph^2]\Bigr)\notag\\
&\qquad\qquad +[\rho\frac{d}{d\rho}(\mathscr{L}^{\rho})_\mu\vert_{\rho =1}]+[\mu\frac{d}{d\mu}\mathscr{L}_\mu]ß ,\label{def(beta)}
\end{align}
where $\langle \mathscr{L}_\mu^{(1)},\ph\rangle$ is the equivalence class of $f\mapsto\langle \mathscr{L}_\mu(f)^{(1)},\ph\rangle$ 
and $\dot{\gamma}_\mu$ and $\dot{\lambda}_\mu$ are the derivatives of $\gamma_\mu(\rho)$ and $\lambda_\mu(\rho)$ 
with respect to $\log\rho$ at $\rho=1$. This formula is 
less complicated as it seems, since $-\dot{\gamma}_\mu\,[(\partial\ph)^2]+\dot{\lambda}_\mu\, m^2\,[\ph^2]$ subtracts precisely the 
$[(\partial\ph)^2]$- and $m^2[\ph^2]$-term of $[B\circ\mathscr{L}_\mu]$.

\section{Examples}
After clarification of the general structure we now want to compute in our framework  renormalization group transformations 
for special examples to low orders. The Gell-Mann Low cocycle $Z(\rho)$ is completely determined by its generator $B$ (\ref{def(B)}).
It can be obtained by differentiating the scaled S-matrix with respect to the scaling parameter,
\[
\rho\frac{d}{d\rho}\vert_{\rho=1}(\sigma_{\rho}\circ S\circ\sigma_{\rho}^{-1})(V)=\rho\frac{d}{d\rho}\vert_{\rho=1}(S\circ Z(\rho))(V)=S^{(1)}(V)B(V) \ ,
\]
where the linear map $S^{(1)}(V)$ is invertible in the sense of formal power series in $\hbar$ 
since its zeroth order term is the pointwise product with $e^{V}$.
 
The computation of $B^{(n)}$ then amounts to  differentiating (\ref{Z-induction}), i.e.
\[
Z^{(n)}(\rho)=\sigma_\rho\circ S^{(n)}\circ \sigma_{\rho}^{-1}-(S\circ Z_{n-1}(\rho))^{(n)}\ ,
\]
where $Z_{n-1}(\rho)$ is given in terms of $\{Z^{(k)}(\rho)\,|\,k\leq n-1\}$ by (\ref{Z-subdiag}). Namely, since $-(S\circ Z_{n-1}(\rho))^{(n)}(V)$
subtracts the contributions coming from the violation of homogeneous scaling of {\it sub}diagrams, one has to compute 
only the contribution of those diagrams whose freedom of normalization is localized on the total diagonal in $\Mk^n$
(cf.~formulas (4.16-17) in \cite{DFret}).

Representing the abstract functionals $A\in\mathcal{A}$ by the explicit functionals $\alpha_{H}(A)\in\mathcal{F}$,
we have to take into account that $\alpha_{H^{\mu}}$ does not commute with the scaling transformations. 
With $S_{\mu}\doteq\al_{H^\mu}\circ S\circ \al_{H^\mu}^{-1}$ and $B_{\mu}\doteq\al_{H^\mu}\circ B\circ \al_{H^\mu}^{-1}$ we obtain 
\begin{align*}
\rho\frac{\partial}{\partial\rho}\vert_{\rho=1}(\sigma_{\rho}\circ S_{\mu}\circ\sigma_{\rho}^{-1})(V)-
\mu\frac{\partial}{\partial\mu}S_{\mu}(V)&= \rho\frac{d}{d\rho}\vert_{\rho=1}(\sigma_{\rho}\circ S_{\rho^{-1}\mu}\circ\sigma_{\rho}^{-1})(V)\\
&=
S_{\mu}^{(1)}(V)B_{\mu}(V) 
\end{align*}
for $V\in\mathcal{F}_{loc}$.
Again, $B^{(n)}_{\mu}(V)$ is obtained directly by omitting the contributions  
coming from the inhomogeneous scaling of subdiagrams.

In terms of the power series expansion of $S$ with respect to $V$, the $\mu$-derivative of the 
$n$th order $S_{\mu}^{(n)}$ can be computed by using
$\mu\frac{\partial}{\partial\mu}\alpha_{H^\mu}=2\hbar\,\Gamma_v\circ \alpha_{H^\mu}$ (\ref{def(alpha_H)},\ref{def(v)}) and 
$\frac{\delta S^{(n)}_\mu}{\delta\varphi}=0$:
\begin{eqnarray}
\mu\frac{\partial}{\partial\mu}S_{\mu}^{(n)}&=&\mu\frac{\partial}{\partial\mu}\,\alpha_{H^\mu}\circ S^{(n)}\circ (\alpha_{H^\mu}^{-1})^{\otimes n}\notag\\
&=&2\hbar\, \Bigl(\Gamma_v\circ  S_\mu^{(n)}- S_\mu^{(n)}\circ\sum(\mathrm{id}\otimes...\otimes\Gamma_v\otimes...\otimes\mathrm{id})\Bigr)\notag\\
&=&2\hbar\,\, S_{\mu}^{(n)}\circ  \sum_{i\ne j}\Gamma_v^{ij}	
\end{eqnarray}
with $\Gamma_v^{ij}\doteq\frac12\int dxdy\, v(x,y)\frac{\delta^2}{\delta\varphi_i(x)\delta\varphi_j(y)}$ as a 
functional differential operator on $\mathcal{F}(\Mk)^{\otimes n}$.

Let 
\[
B_\mu(V)=\sum_{k=2}^{\infty}\frac{1}{k!}\,B_\mu^{(k)}(V^{\otimes k}) \ ,\ V\in\mathcal{F}_\mathrm{loc}(\Mk)
\]
be the power series expansion of $B_{\mu}$. 
Due to $\frac{\delta\, Z(\rho)}{\delta\varphi}=0$ we have the commutation relation 
\[
\frac{1}{n!}\frac{\delta^n}{\delta\ph^n}\circ B_\mu(V)=\sum_{k=1}^\infty\sum_{n_1+\ldots +n_k=n}\frac{1}{k!}\,
B_{\mu}^{(k)}\circ \left(\frac{1}{n_1!}\frac{\delta^{n_1}}{\delta\ph^{n_1}}\otimes\cdots\otimes\frac{1}{n_k!}\frac{\delta^{n_k}}{\delta\ph^{n_k}}\right)(V^{\otimes k})\ .
\] 
The Taylor expansion of $B_\mu(V)$ around $\ph=0$ is therefore determined by the values of $B_{\mu}^{(k)}(V^{(n_1)}\otimes\cdots\otimes V^{(n_k)})$ at $\ph=0$.  
To obtain these values we have to compute the violation of homogeneous scaling of the 
corresponding renormalized time ordered products at $\ph=0$. Mostly this is done by using 
momentum space techniques. A rigorous computation in our framework is, however, easier in configuration space which also has the advantage to simplify the extension to curved spacetimes. Our method is related to differential renormalization 
\cite{FJL}. Further useful renormalization procedures in $x$-space are dimensional regularization \cite{BG}, 
different kinds of analytic renormalization \cite{Guett,HYM}
and in case of 2-point functions a method relying on the K\"allen-Lehmann representation \cite{DFret}.

In order to simplify the formulas that will appear in the next two subsections, we use the notations $\phi^n \doteq \varphi^n/n!$, $n\in\NN$,  and $(\partial\phi)^2\doteq \partial_\nu\varphi\partial^{\nu}\varphi/2$.
\subsection{$\beta$-function for the $\varphi^3$-interaction in 6 dimensions}
As a first example we discuss the $\phi^3$ interaction in 6 dimensions. Let 
\[
\mathscr{L}_\mu(f)=\frac{ig}{\hbar}\al_{H^{\mu}}^{-1}\int dx\, f(x)\ {\phi(x)^3}
\] 
be a generalized Lagrangian with a coupling constant $g\in\R$ . The orbit under the renormalization group is contained in the subspace generated by $\phi^3$, $\phi^2$, $(\partial\phi)^2$, $\phi$ and 1. Since we are interested only in equivalence classes of Lagrangians we may ignore the constant terms. Moreover, we may also ignore the linear terms, since they do not influence the action of the renormalization group on the other terms.  
In second order in the Lagrangian the only contribution comes from $\phi^3$ and yields a term of second order in $\phi$. This contribution is therefore determined 
by 
\begin{eqnarray*}
B_\mu^{(2)}\left({\phi^3(x_1)}\otimes{\phi^3(x_2)}\right)& = & \,\left.B_\mu^{(2)}\left({\phi^2(x_1)}\otimes{\phi^2(x_2)}\right)\right\vert_{\phi=0}\,\,\phi(x_1)\phi(x_2)\ ,\\
\left.B_\mu^{(2)}\left({\phi^2(x_1)}\otimes{\phi^2(x_2)}\right)\right\vert_{\phi=0}& = & \rho\frac{d}{d\rho}\vert_{\rho=1}\rho^8 
t^{\rho^{-1}\mu}_{\rho^{-1} m}(\phi^2,\phi^2)(\rho (x_1 -x_2))\ ,\end{eqnarray*}
where $t^{\mu}_m(\phi^2,\phi^2)$ is an extension of $(\hbar H_{F,m}^{\mu})^2/2$ (to be considered as a distribution on $\mathcal{D}(\Mk^2\setminus \Delta_2))$ 
to an everywhere defined distribution with the same scaling degree. Moreover, exhibiting the dependence on the mass, we write 
$$(H_{F,m}^{\mu})^2/2=D_F^2/2+m^2 D_F\cdot\frac{\partial}{\partial m^2}H_{F,m}^{\mu}(m=0)+R^{\mu}_m$$ where the scaling degree of $R^{\mu}_m$ is smaller 
than the dimension, such that there is a unique extension with the same scaling degree. Since $\rho^8 R^{\rho^{-1}\mu}_{\rho^{-1}m}(\rho x)$ (with $x\equiv x_1 -x_2$) is independent from $\rho$, only the two first terms contribute to $B_{\mu}$.

The extension $t^{\mu}_m$ can be written in the form
\[
t^{\mu}_m=t_0+m^2t_1+R^{\mu}_m
\]     
where $t_0,t_1$ are extensions of $\hbar^2 D_F^2/2$ and $\hbar^2 D_F\cdot\frac{\partial}{\partial m^2}H^{\mu}_F(m=0)$, respectively, which scale almost homogeneously.

From Appendix A (\ref{Feynmanprop-6}), we get $D_F(x)=\frac{1}{4\pi^3(x^2-i\epsilon)^2}$ and 
$\frac{\partial}{\partial m^2}H^{\mu}_F(m=0)=\frac{1}{2^4\pi^3(x^2-i\epsilon)}$. We can now use the general results of Appendix C and conclude
\[
\rho\frac{d}{d\rho}|_{\rho=1}(\rho^8 t_0(\rho x)+m^2\rho^6 t_1(\rho x))|_{\rho=1}
=\hbar^2(a_0\square\delta(x)+a_1m^2\delta(x)) 
\]
where  $a_0=\frac{i}{2^7\cdot 3\pi^3}$ and $a_1=\frac{i}{2^6\pi^3}$. 

We now turn to terms of 3rd order in the Lagrangian. There is no term of second order in $\phi$, 
and there is exactly one term in third order in $\phi$, corresponding to the triangle diagram. We have to calculate
\[
B^{(3)}_\mu\left(\otimes_{j=1}^3{\phi^3(x_j)}\right)=
\,\left.B^{(3)}_\mu\left(\otimes_{j=1}^3 {\phi^2(x_j)}\right)\right\vert_{\phi=0}\,\,\phi(x_1)\phi(x_2)\phi(x_3) \ .
\] 
Again the $\mu$-dependent part is regular and can thus be absorbed in the $\rho$ derivative. We obtain 
\[
B^{(3)}_\mu\left(\otimes_{j=1}^3{\phi^2(x_j)}\right)\vert_{\phi=0}=
\rho\frac{d}{d\rho}|_{\rho=1}\rho^{12} t^{\rho^{-1}\mu}_{\rho^{-1}m}(\phi^2,\phi^2,\phi^2)(\rho x,\rho y) 
\]
where $x\equiv x_1-x_2\,,\, y\equiv x_2-x_3$.
$t_m^{\mu}(\phi^2,\phi^2,\phi^2)(x,y)$ is an almost homogeneous extension of $\hbar^3 H_F^{\mu}(x)H_F^{\mu}(y)H_F^{\mu}(x+y)$. In the Taylor expansion 
with respect to the mass only the leading term has a nonunique extension which can give rise to a contribution to $B_\mu$. Hence, 
we have to determine the distribution
\[
B^{(3)}_\mu\left(\otimes_{j=1}^3{\phi^2(x_j)}\right)\vert_{\phi=0}=\hbar^3(\partial_{x,\mu} x^\mu+\partial_{y,\mu}y^\mu)D_F(x)D_F(y)D_F(x+y)\ . 
\]
As before this is the divergence of a distribution with scaling degree less than the dimension of space. 
It must have the form $\hbar^3 a_2\delta(x)\delta(y)$ with $a_2\in\C$. 

To compute this number, we use the explicit form of the Feynman propagator $D_F$ and the method of Feynman parameters:
\[
\frac{1}{b_1b_2...b_6}=5!\int_0^\infty\frac{dz_1dz_2...dz_6\,\delta(1-(z_1+z_2+...+z_6))}{(b_1z_1+b_2z_2+...+b_6z_6)^6}\ .
\]
With that we find
\[
a_2\delta(z)=\lim_{\epsilon\downarrow 0}(4\pi^3)^{-3}5! 
\int_{\alpha,\beta,\gamma>0,\alpha+\beta+\gamma=1}d\alpha d\beta\, \alpha\beta\gamma 
\,\frac{\partial}{\partial z^j}z^j\,(\langle z,Gz\rangle-i\epsilon)^{-6}
\]  
where $z=(x,y)$ and $\langle z,Gz\rangle=\alpha x^2+\beta y^2+\gamma (x+y)^2$.

Up to a linear coordinate transformation on $\R^{12}$ which brings the quadratic form $G$ into the standard form on $\R^{2,10}$, the distribution $(\langle z,Gz\rangle-i\epsilon)^{-6}$ is of the form treated in Appendix C. We conclude that
\[
\frac{\partial}{\partial z^j}z^j\,(\langle z,Gz\rangle-i\epsilon)^{-6}=|S^{11}|\,|\det G|^{-\frac12}\delta(z)
\]
Thus we find
\[a_2=(4\pi^3)^{-3}5!\cdot\frac{1}{10}\frac{\pi^6}{6}I= \frac{1}{2^5\pi^3}I\]
with the integral over the Feynman parameters
\[I=\int_{\alpha,\beta,\gamma>0,\alpha+\beta+\gamma=1}d\alpha d\beta\, \alpha\beta\gamma\,|\det{G}|^{-\frac12}\ .\]
The determinant of $G$ is $(\alpha\beta+\beta\gamma+\gamma\alpha)^6$. To compute this integral we substitute $\alpha=\lambda\kappa$, $\beta=(1-\lambda)\kappa$ with $\lambda,\kappa\in(0,1)$.
We find
\begin{multline*}
I=\int_0^1d\lambda\int_0^1d\kappa \, \frac{\lambda(1-\lambda) \kappa^3(1-\kappa)}{(\lambda(1-\lambda)\kappa^2+\kappa(1-\kappa))^3}\\
=\int_0^1d\lambda\int_0^1d\kappa\,\frac{\lambda(1-\lambda)(1-\kappa)}{(\lambda(1-\lambda)\kappa+(1-\kappa))^3}
\ . 
\end{multline*}
The integral over $\kappa$ turns out to be independent of $\lambda\in(0,1)$ and has the value $\frac12$. Thus we finally obtain
\[a_2=\frac{1}{2^6\pi^3} \ .\]
We arrive at the action of $B_{\mu}$ on the interaction classes up to third order
\[
[B_\mu\circ g{\phi^3}]=\left[-\hbar^2\frac{ig^2}{3\cdot 2^7\pi^3}{(\partial\phi)^2}+
\hbar^2\,m^2\,\frac{ig^2}{2^6\pi^3}{\phi^2}+\hbar^3\frac{g^3}{2^6\pi^3}{\phi^3}\right] \ .
\]
Using the formula (\ref{def(beta)}) for the $\beta$-function and exhibiting the factor $\frac{i}{\hbar}$ in the interaction, we get in lowest 
nontrivial order (rewriting everything in terms of the original notation for fields)
\[
\frac{\hbar}{i}\beta\left(\frac{i}{\hbar}g\frac{[(\ph^3)_\mu]}{3!}\right)= -\frac{3\,\hbar g^3}{2^8\pi^3}\frac{[(\ph^3)_\mu]}{3!}
\]
from which we read off the coupling constant renormalization. Our result agrees with formula (3.4.64)
in \cite{Muta}; the negative sign exhibits asymptotic freedom. 
\subsection{Renormalization group flow for the $\varphi^4$-interaction in 4 dimensions}
As another example we study the $\phi^4$ interaction in 4 dimensions. We may restrict ourselves to renormalizations 
which respect the symmetry $\phi\to -\phi$. The orbit in the renormalization group is contained in the subspace 
generated by $\phi^4,\phi^2,(\partial\phi)^2$ and 1, where the constant terms can again be ignored. 
We are going to determine the renormalization goup flow {\it completely},
i.e.~we will compute $B_{\mu}(\frac{i}{\hbar}(g\,[\phi^4]+am^2\,[\phi^2]+b[(\partial\phi)^2]))$.

In second order in the interaction the
following terms occur
\begin{eqnarray*}
B_{\mu}^{(2)}\left({\phi(x)^4}\otimes{\phi(y)^4}\right)&=&
\left.B_{\mu}^{(2)}\left({\phi(x)^3}\otimes{\phi(y)^3}\right)\right|_{\phi=0}\phi(x)\phi(y)\\
&&\qquad+\left .B_{\mu}^{(2)}\left({\phi(x)^2}\otimes{\phi(y)^2}\right)\right|_{\phi=0}{\phi(x)^2}{\phi(y)^2} \ ,
\end{eqnarray*}
\[
B_{\mu}^{(2)}\left({\phi(x)^4}\otimes{\phi(y)^2}\right)=
\left.B_{\mu}^{(2)}\left({\phi(x)^2}\otimes{\phi(y)^2}\right)\right|_{\phi=0}{\phi(x)^2} \ ,
\]
\[
B_{\mu}^{(2)}\left({\phi(x)^4}\otimes{(\partial\phi(y))^2}\right)=
\left.B_{\mu}^{(2)}\left({\phi(x)^2}\otimes{(\partial\phi(y))^2}\right)\right|_{\phi=0}{\phi(x)^2} \ .
\]
The computation of the term $\left.B_{\mu}^{(2)}\left({\phi(x)^2}\otimes{\phi(y)^2}\right)\right|_{\phi=0}$ 
proceeds as in the case of $\phi^3_6$ (taking now the explicit expression for $D_F$ from (\ref{Feynmanprop-4})) and yields
\begin{equation}
\left.B_{\mu}^{(2)}\left({\phi(x)^2}\otimes{\phi(y)^2}\right)\right|_{\phi=0}=
\frac{-i\hbar^2}{2^4\pi^2}\delta(x-y) \ .\label{B-fisch}
\end{equation}
The term $\left.B_{\mu}^{(2)}\left({\phi(x)^2}\otimes{(\partial\phi(y))^2}\right)\right|_{\phi=0}$
is of the form
\[
\left.B_{\mu}^{(2)}\left({\phi(x)^2}\otimes{(\partial\phi(y))^2}\right)\right|_{\phi=0}=
\hbar^2\left(a_1\square\delta(x-y)+b_1m^2\delta(x-y)\right) \ .
\]
Inserting this form into the expression for $B_{\mu}^{(2)}\left({\phi(x)^4}\otimes{(\partial\phi(y))^2}\right)$, one observes that the term proportional to $\square\delta$ produces a total derivative in the interaction Lagrangian and may therefore be ignored. The term of order $m^2$ arises from
\[
b_1\delta(x)=\partial_{\nu}\left(x^{\nu}(\partial_{\lambda}D_F)(x)\partial^{\lambda}\frac{\partial}{\partial m^2}H_{F,\mu}(m^2=0)(x)\right) \ .
\] 
In $d=4$ dimensions we have 
\[
\partial^{\lambda}\frac{\partial}{\partial m^2}H_{F,\mu}(m^2=0)(x)=\frac{1}{2^3\pi^2}\frac{x^{\lambda}}{x^2-i\eps}
\]
and 
\[
\partial_{\lambda}D_F(x)=\frac{1}{2\pi^2}\frac{x_{\lambda}}{(x^2-i\eps)^2} 
\]
(both formulas can be read off from (\ref{Feynmanprop-4})), thus
\[
\partial_{\lambda}D_F(x)\partial^{\lambda}\frac{\partial}{\partial m^2}H_{F,\mu}(m^2=0)(x)=\frac{1}{2^4\pi^4}\frac{1}{(x^2-i\eps)^2}
\]
and we are back to the case treated in Appendix C. We obtain
\[
b_1=\frac{-i}{2^3\pi^2} \ .
\]
The most interesting case is the contribution of $\left.B_{\mu}^{(2)}\left({\phi(x)^3}\otimes{\phi(y)^3}\right)\right|_{\phi=0}$. We have
\[
\left.B_{\mu}^{(2)}\left({\phi(x)^3}\otimes{\phi(y)^3}\right)\right|_{\phi=0}=\rho\frac{d}{d\rho}t_{\rho^{-1}m}^{\rho^{-1}\mu}(\phi^3,\phi^3)(\rho(x-y)) 
\]
where $t_m^{\mu}(\phi^3,\phi^3)$ is an extension of
$\frac{(\hbar\, H_{F,m}^{\mu})^3}{3!}$ whose $\mu$-dependence is given in terms of $t_m^{\mu}(\phi^2,\phi^2)$,
\[
\mu\frac{\partial}{\partial\mu}t_m^{\mu}(\phi^3,\phi^3)(x,y)=2\hbar\,v(x-y)t_m^{\mu}(\phi^2,\phi^2)(x,y)
\]
 
We compute for $x\ne 0$
\[
\frac{(H^{\mu}_{F,m})^3(x)}{3!}=\frac{D_{F}^3(x)}{3!}+ m^2\frac{D_F^2(x)}{2}\frac{\partial}{\partial m^2}H_{F,m}^{\mu}(m^2=0)(x)+R_m^{\mu}(x)
\]
where $R$ is regular and scales homogeneously. The scale dependence of the extensions $t_0$  of $D_{F}^3/3!$ is of the form treated in Appendix C and is given by
\[
c(t_0)\delta=\frac{i}{2^9\cdot 3 \pi^4}\square\delta
\]

The next leading term in $m^2$ is still singular. Its extension $m^2 t_1^{\mu}$ is $\mu$-dependent, where the $\mu$-dependence as shown above is determined from the $m^2=0$ contribution $t_2$ of $t_m^{\mu}(\phi^2,\phi^2)$,
\[
\mu\frac{\partial}{\partial\mu}t_1^{\mu}=2\hbar\,\frac{\partial\,v}{\partial m^2}\vert_{m^2=0}\,t_2 \ .
\]
We now compute the partial $\rho$-derivative (with fixed $\mu$) of $t_1^{\mu}$.
This is the divergence of $x^{\nu} t_1^{\mu}$ which is the unique extension of 
$\hbar^3\,x^{\nu}\frac{D_F^2}{2}\left(\frac{\partial}{\partial m^2}H_{F,m}^{\mu}(m^2=0)(x)\right)$. 

From Appendix A (\ref{Feynmanprop-4}) we get
\[\frac{\partial}{\partial m^2}H_{F,m}^{\mu}(m^2=0)(x)=\frac{1}{2^4\pi^2}\log(-\mu^2(x^2-i\eps))+F(0)\ ,\quad\> F(0)=\frac{2\,C-1}{4\pi}.
\]

We use
\[\partial^{\nu}\frac{\log(-\mu^2(x^2-i\epsilon))}{x^2-i\epsilon}=2x^{\nu}\frac{1-\log(-\mu^2(x^2-i\epsilon))}{(x^2-i\epsilon)^2}\]
and the fact that
\[
-\square \frac{\log(-\mu^2(x^2-i\epsilon)}{x^2-i\epsilon}
\]
is an extension of $4\cdot(4\pi^2)^2D_F(x)^2$. $t_2$ is an extension of $\frac{\hbar^2}{2}\,D_F^2$ which we may parametrize by a real parameter $\tau$,
\[
t_2(x)=-\frac{\hbar^2}{2^7\,\pi^4}\,\square\frac{\log(-\tau^2(x^2-i\epsilon))}{x^2-i\epsilon}\ .
\] 
We also use $\frac{\partial\,v}{\partial m^2}\vert_{m^2=0}=\frac{1}{2^4\,\pi^2}$ which results from (\ref{def(v)},\ref{Feynmanprop-4}). In 
$(\partial_{\nu}x^{\nu}-\mu\frac{\partial}{\partial\mu}) t_1^{\mu}$
the two extensions cancel up to a multiple of a $\delta$-function, and we finally obtain
\begin{eqnarray*}
B_{\mu}^{(2)}\left({\phi(x)^3}\otimes{\phi(y)^3}\right)|_{\phi=0}&=& i\hbar^3\,\Bigl(\frac{1}{2^9\cdot 3\,\pi^4}\,\square\delta(x-y)\Bigr.\\
&-&\Bigl. \left(\frac{1}{2^8\,\pi^4}\,(1+\log\frac{\mu^2}{\tau^2})+\frac{F(0)}{2^4\,\pi^2}\right)\,m^2\,\delta(x-y) \Bigr)\ .
\end{eqnarray*}
We conclude that the action of $B_{\mu}$ on interaction classes is to second order given by
\begin{align*}
&\left[B^{(2)}_{\mu}\circ\left(\left(g\,{\phi^4}+a\,m^2\,{\phi^2}+b\,{(\partial\phi)^2}\right)^{\otimes 2}\right)\right]\\
&\quad=-\frac{i\,\hbar^2\,g^2\,3}{2^4\,\pi^2}\,{[\phi^4]}-\frac{i\,\hbar^3\,g^2}{3\cdot 2^9\,\pi^4}\,{[(\partial\phi)^2]}\\
&\qquad\quad -i\left(\frac{\hbar^2\,a\,g}{2^4\,\pi^2}+\frac{\hbar^2\,b\,g}{2^3\,\pi^2}+\hbar^3\,g^2\,
\left(\frac{1}{2^8\,\pi^4}\,(1+\log\frac{\mu^2}{\tau^2})+\frac{F(0)}{2^4\,\pi^2}\right)\right)\,m^2\,{[\phi^2]}\ .
\end{align*}

In the argument of the $\beta$-function we omit the $[(\partial\phi)^2]$- and $m^2\,[\phi^2]$-term, 
since $\beta$ is defined after having absorbed these terms by the wave function and mass renormalization. Using
$[\mu\frac{d}{d\mu}(\phi^4)_\mu]=-2\hbar\,[\Gamma_v((\phi^4)_\mu)]=-2\hbar\,[\Gamma_v\phi^4]=-\frac{\hbar}{16\pi^2}m^2[\phi^2]$
(see (\ref{def(v)},\ref{Feynmanprop-4})) we obtain 
\begin{equation*}
\frac{\hbar}{i}\beta\left(\frac{ig}{\hbar}{[(\phi^4)_\mu]}\right)
=\frac{g^2\,\hbar\,3}{2^4\,\pi^2}\,{[(\phi^4)_\mu]}-\frac{g\,\hbar}{2^4\,\pi^2}
\,m^2\,{[\phi^2]}+...\ \ .
\end{equation*}
where the dots stand for terms of third or higher orders in $g$.
From
\[
\frac{\dot{\gamma}_\mu}{2}\,\langle {g}(\phi^4)_\mu^{(1)},\phi\rangle=\frac{\dot{\gamma}_\mu\,g}{2}
\,\left(4\,[(\phi^4)_\mu] +2\hbar\,[\Gamma_{H^\mu}\phi^4]\right)
\]
we get a further $[\phi^2]$-term, which is of order $g^3$. Note that if we apply the definition (\ref{def(beta)}) of the $\beta$-function
to an interaction which is $\sim [\phi^4]$, we get a result which is  $\sim [\phi^4]$.

\section{Conclusions}
Quantum field theory on generic curved backgrounds requires a revision of the standard methods of perturbative quantum field theory; in particular the dependence on the choice of a distinguished state (the ''vacuum'') would introduce an unwanted nonlocal feature and has to be avoided in order to remain in agreement with the principle of general covariance. This program could be successfully performed \cite{BF,HW1,HW2} by the use of the following ingredients:
\begin{itemize}
\item Algebras of observables are directly constructed without a detour via expectation values in distinguished states.
\item Techniques of microlocal analysis replace momentum space techniques.
\item Dimension full parameters (e.g. the mass) are treated as expansion parameters. 
\end{itemize}  
In the present paper we analyzed the consequences of this approach for standard quantum field theory and compared our formalism with other formalisms, in particular with the method of renormalization by the flow equation \cite{Polchinski,Salmhofer2}. 

The independence of the formalism on the choice of the mass required the replacement of the vacuum two point function by a so-called Hadamard function which differs from it by a smooth function of position. In even dimensions the Hadamard function depends on an additional mass parameter, whereas in odd dimensions it is unique. In particular it also exists in 2 spacetime dimensions.

The (off shell) observables were represented as functionals on a space of smooth field configurations. Several algebraic structures on the space of observables were introduced: The (classical) product by pointwise multiplication, the (quantum) product as a $*$-product involving the Hadamard function and the time ordered product by which the interacting theory was constructed inside the (off shell) algebra of the free theory. Since all these algebraic structures involve only the functional derivatives of the observables with respect to the field, the formalism is not restricted to polynomial functionals, as long as one remains in the realm of formal power series.

Renormalization in this framework consists in an extension of the time ordered product to the more singular local functionals for which we gave a new intrinsic definition. 
The extension can be done by the methods of the St\"uckelberg-Bogoliubov-Epstein-Glaser approach \cite{EG}, but is not unique. The nonuniqueness is described by a group of transformations on the space of local functionals, which is the renormalization group in the sense of St\"uckelberg and Petermann \cite{SP}.  

A comparison with the method of flow equations can be reached by approximating the Hadamard function by a more regular family of functions labeled by a regularization parameter $\Lambda$. The effective potential as a function of $\Lambda$ can be defined and is shown to satisfy the flow equation. As a consequence of the existence of extensions in the sense of the Epstein-Glaser method it immediately follows that there exist appropriate counter terms which guarantee the convergence of the effective potentials. The somewhat cumbersome estimates in the flow equation method which up to now complicated a generalization of the method to generic Lorentzian spacetimes (see \cite{Kopper} for Minkowski space) are not required. Moreover, the choice of regularization is completely arbitrary. One may, in particular, make a specific choice of a parameter such that the effective potential becomes a 
meromorphic function of it, with a possible pole at the removal of the cutoff. This might lead to explicit choices of extensions as e.g. minimal subtraction for dimensional regularization. 

Of particular interest is the behaviour of the theory under scaling. We found a purely algebraic analogue of the Callan-Symanzik equation, much in the spirit of the Buchholz-Verch approach to an intrinsic renormalization group within axiomatic algebraic quantum field theory \cite{BV}, but technically quite different.
As a matter of fact our analysis is completely free of any dependence on the mass of the theory. In standard perturbation theory a corresponding observation was made in the context of dimensional regularization by Collins \cite{Collins}, but there the reasons for this effect remained mysterious.

The Epstein-Glaser method relies on coupling constants which are test functions with compact support. This avoids all infrared problems during the construction but leads to the problem of the adiabatic limit in which the test functions approach constant functions.
In general, all infrared problems now could reappear. But exploiting the method of the algebraic adiabatic limit \cite{BF} the construction of the algebra of observables can be done directly, and in this paper we show how this method can also be used to define the renormalization group and the beta function in the adiabatic limit. In particular, the beta function turns out to be state independent (and, as noted previously \cite{HW3}, independent of the topological features of spacetime). 

We compared our findings with the standard definition of the beta function, present explicit calculations within our framework for $\ph^3_6$ and $\ph^4_4$ and obtain agreement with the literature.

We did not yet enter a detailed comparison of our method with the BPHZ method and its modern version in terms of the Connes-Kreimer theory \cite{Connes1,Connes2}. In spite of the fact that both methods are known to be equivalent, the involved combinatorics of renormalization is quite different, and understanding the relations requires additional work. We hope to return to this problem in a future publication.     

\begin{appendix}
\section{Determination of the Hadamard function $H$}\label{Hadamard}
The Wightman 2-point function $\Delta^+$ in $d\ge2$ dimensions for $m^2>0$ can be expressed
in terms of modified Bessel functions
\begin{equation}
\Delta^+_m(x)=(2\pi)^{-\frac{d}{2}}\, m^{d/2-1}|x^2|^{\frac{2-d}{4}}K_{d/2-1}(\sqrt{m^2|x^2|}) \label{Delta^+}
\end{equation}
for spacelike arguments $x$.
To obtain the Hadamard function, we have to add a smooth Lorentz invariant solution of the Klein-Gordon equation such that the sum is a smooth function of 
$m^2$. Each Lorentz invariant solution $F$ is for spacelike arguments of the 
form $F(x)=|x^2|^{\frac{2-d}{4}}G(\sqrt{m^2|x^2|})$ where $G$ satisfies the modified Bessel equation of order $d/2-1$,
\begin{equation}
G''(y)+\frac{1}{y}G'(y)-\left(1+\frac{(\frac{d}{2}-1)^2}{y^2}\right)G(y)=0\ .\label{Bessel}
\end{equation} 
The solutions of this differential equation are linear combinations either of $I_{\frac{d}{2}-1}(y)$ and
$I_{1-\frac{d}{2}}(y)$ (if $d$ is odd) or of $I_{\frac{d}{2}-1}(y)$ and $K_{\frac{d}{2}-1}(y)$ (if $d$ is even).
In both cases smoothness in $x$ at $x=0$ then implies that $F$ is a multiple of 
$|x^2|^{\frac{2-d}{4}}I_{\frac{d}{2}-1}(\sqrt{m^2|x^2|})$.
Modified Bessel functions of noninteger order $\nu$ satisfy the relation
\begin{equation}
K_{\nu}=\frac{\pi}{2\sin\nu\pi}(I_{-\nu}-I_{\nu}) \ .\label{K=I-I}
\end{equation} 
With that, in odd dimensions $d$, we obtain the unique Hadamard function
\begin{equation}\label{Hodd}
H_m(x)=\frac{1}{4\sin(\frac{d}{2}-1)\pi}\,(2\pi)^{\frac{2-d}{2}}\,m^{d/2-1}|x^2|^{\frac{2-d}{4}}
I_{1-d/2}(\sqrt{m^2|x^2|})
\end{equation}
for $ x^2<0$. Namely, since $I_{\nu}(y)$ is of the form $y^{\nu}F(y^2)$ with an entire analytic 
function $F$, $H_m$ is a smooth function of $m^2$.

In even dimension $d$ we introduce a parameter $\mu$ with the dimension of a 
mass and consider the family of functions
\begin{eqnarray}\label{Hoddz}
H_m^{\mu,z}(x)&=&\frac{\mu^{-z}}{4\sin(\frac{d+z}{2}-1)\pi}\,(2\pi)^{\frac{2-(d+z)}{2}}\notag\\
&\cdot& m^{(d+z)/2-1}|x^2|^{\frac{2-(d+z)}{4}}
I_{1-(d+z)/2}(\sqrt{m^2|x^2|})
\end{eqnarray} 
where $x^2<0$ is assumed. The factor $\mu^{-z}$ is needed in order that $H_m^{\mu,z}(x)$ has the dimension $[m^{d-2}]$ (as required for a
2-point function in $d$ dimensions). 
For the same reasons as for $H_m$ (\ref{Hodd}), this 'dimensionally regularized Hadamard function'
is smooth in $m^2$ and differs from $\Delta^{+\,(d+z)}_m(x)$ (i.e.~(\ref{Delta^+}) with $d$ replaced by $d+z$) by a smooth 
Lorentz invariant function which 'solves the Klein-Gordon equation in $(d+z)$-dimensions' (i.e.~it solves (\ref{Bessel}) for $(d+z)$).
By using (\ref{K=I-I}) we express $I_{1-(d+z)/2}$ in terms of $I_{(d+z)/2-1}$ and $K_{(d+z)/2-1}$. The limit $z\to 0$
exists for the $K_{(d+z)/2-1}$-term and gives $\Delta^+_m(x)$. But the $I_{(d+z)/2-1}$-term 
is meromorphic in $z$ with a simple pole at $z=0$. Since the residuum is smooth in $m^2$ and, with respect to $x$, a smooth 
Lorentz invariant solution of the Klein-Gordon equation\footnote{Note that the residuum of $H_m^{\mu,z}$ at $z=0$
is $\sim |x^2|^{\frac{2-d}{4}}I_{1-d/2}(\sqrt{m^2|x^2|})$ and, hence, not smooth in $x$.}, we may subtract the pole term. Taking then the limit $z\to 0$
we get the Hadamard function
\begin{equation}\label{Heven}
H^{\mu}_m(x)=\Delta^+_m(x)+\frac{(-1)^{\frac{d}{2}}}{2(2\pi)^{\frac{d}{2}}}\,
\log{\frac{\mu^2}{m^2}}\, m^{d/2-1}|x^2|^{\frac{2-d}{4}}I_{d/2-1}(\sqrt{m^2|x^2|})
\end{equation}
for $ x^2<0$, which is unique up to the choice of the parameter $\mu$.

The values of $H$ (any dimension) for arbitrary $x$ and $m^2\in\R$ are obtained by replacing $|x^2|$ by $-(x^2-ix^0 0)$ and then by symmetrizing w.r.t. $x$.

The corresponding Feynman propagator $H^\mu_F$ is defined by
\[
H^\mu_F(x)=\theta(x^0)H^\mu(x)+\theta(-x^0)H^\mu(-x)\ .
\]
We conclude that the explicit exression for $H^\mu_F$ is obtained from (\ref{Hodd}) and (\ref{Heven}), respectively,
by replacing $|x^2|$ by $-(x^2-i 0)$: e.g.~in even dimensions it results
\[
H^\mu_F(x)=\frac{m^{d-2}}{(2\pi)^{\frac{d}{2}}\, y^{\frac{d}{2}-1}}\,\left(K_{\frac{d}{2}-1}(y)+(-1)^{\frac{d}{2}}\,
\log{\frac{\mu}{m}}\,\,I_{\frac{d}{2}-1}(y)\right)\ ,
\]
where $y\doteq\sqrt{-m^2(x^2-i0)}$. In the main text the following 
formulas are used for explicit computations in $d=4$ and $d=6$ dimensions:
\begin{eqnarray}
H_F^{\mu\,(4)}(x)&=&\frac{-1}{4\pi^2(x^2-i0)}\notag\\
&+&{\rm log}(-\mu^2(x^2-i0))\, m^2 f(m^2x^2)+m^2 F(m^2x^2)\ ,\label{Feynmanprop-4}\\
H_F^{\mu\,(6)}(x)&=&\frac{1}{4\pi^3(x^2-i0)^2}+\frac{m^2 f(m^2x^2)}{\pi\,(x^2-i0)}\notag\\
&+&\frac{1}{\pi}\left({\rm log}(-\mu^2(x^2-i0))\, m^4 f'(m^2x^2)+m^4 F'(m^2x^2)\right)\ ,\label{Feynmanprop-6}
\end{eqnarray}
where $f$ and $F$ are realvalued analytic functions.
$f$ and $f'$ can be expressed in terms of the Bessel functions $J_1$ and $J_2$, respectively, namely
\begin{equation}
f(z)\doteq \frac{1}{8\pi^2\sqrt{z}}\>J_1(\sqrt{z})\ ,\quad f(0)=\frac{1}{2^4\,\pi^2}\ ,
\quad f'(z)=\frac{-1}{16\,\pi^2\,z}\>J_2(\sqrt{z})\ ;\label{f}
\end{equation}
and $F$ is given by a power series
\begin{equation}
F(z)\doteq -\frac{1}{4\pi}\sum_{k=0}^\infty\{\psi(k+1)+\psi(k+2)\}
\frac{(-z/4)^k}{k!(k+1)!}\ , \quad F(0)=\frac{2\,C-1}{4\pi}\ ,
\end{equation}
where $C$ is Euler's constant and the Psi-function is related to the Gamma-function by
$\psi(x)\doteq\Gamma^\prime (x)\, /\, \Gamma (x)$.
\section{Additivity of $Z$}
\label{localityZ}

In this Appendix we derive the additivity relation  (\ref{Zloc1}) of the renormalization group transformation $Z$ under the assumption that 
$Z$ satisfies the simple additivity relation $Z(A+C)=Z(A)+Z(C)$ for $\supp A\cap\supp C=\emptyset$.
We use the fact that $Z(A+\lambda B+C)$ is determined within perturbation theory by its derivatives with respect to $\lambda$ at $\lambda=0$.
We prove that the $n$th derivative on both sides coincide for all $n$. 

Let $A,B,C\in \A_{\sst{loc}}$ with $\supp A\cap\supp C=\emptyset$. From Lemma 3.2 we conclude that $B$ may be written as a sum of $N>n$ terms, 
$B=\sum_{i=1}^NB_i$ (where $B_i\in \A_{\sst{loc}}$)
such that all subsets $I\subset\{1,\ldots,N\}$ with at most $n$ elements admit a decomposition 
$I=I_1\cup I_2$, $I_1\cap I_2=\emptyset$ such that
$(\supp A\cup\bigcup_{i\in I_1}\supp B_i)\cap(\supp C\cup\bigcup_{j\in I_2}\supp B_j)=\emptyset$.

Let $B(\lambda)=\sum_{i=1}^N\lambda_iB_i$ for $\lambda=(\lambda_1,\ldots,\lambda_N)\in\R^N$.
We prove that for every multiindex $\alpha=(\alpha_1,\ldots,\alpha_N)$ with $|\alpha |=\sum_{i=1}^N\alpha_i\le n$ the derivative
\[\partial^{\alpha}_{\lambda}\left(Z(A+B(\lambda)+C)-Z(A+B(\lambda))+Z(B(\lambda))-Z(B(\lambda)+C)\right)\]
vanishes at $\lambda=0$. Let $I=\{i\in\{1,\ldots,N\}\, |\, \alpha_i\ne 0\}$. 
We choose a decomposition $I=I_1+I_2$ as described above. Let $B_{I_1}(\lambda)=\sum_{i\in I_1}\lambda_iB_i$ and  $B_{I_2}(\lambda)=\sum_{j\in I_2}\lambda_jB_j$. Then the derivative above does not change at $\lambda=0$ when we replace $B(\lambda)$ by $B_{I_1}(\lambda)+B_{I_2}(\lambda)$. Due to the assumed support properties we may now use the simple additivity relation and obtain the result that the derivative vanishes. This proves the claim. 
\section{Scaling violations of extensions of homogeneous\\ distributions}\label{Scaling}
The position space renormalization relies crucially on the Theorem on the Extension of Distributions. This theorem goes back 
to Epstein-Glaser \cite{EG} and Steinmann \cite{Stein} and was generalized to differentiable manifolds by Brunetti and 
Fredenhagen \cite{BF}. A refinement for almost homogeneous distributions was obtained by Hollands and Wald \cite{HW2}. 
(Cf. also \cite{GB,DFret}.)
\begin{teo}
Let $t_0\in\Dcal'(\R^d\setminus\{0\})$ such that $(\rho\frac{d}{d\rho})^k\rho^lt_0(\rho\,\,\cdot)=0$ 
for some $k\in\NN$, $l\in\R$. Then there exists an almost homogeneous distribution $t\in\Dcal'(\R^d)$ 
which coincides with $t_0$ outside of the origin. If $l<d$ or $l\not\in\ZZ$, $t$ is unique
and fulfills the scaling relation with the same power $k$ as $t_0$. If $l\in d+\NN_0$, $t$ satisfies the scaling condition
\[
(\rho\frac{d}{d\rho})^{k+1}\rho^lt_0(\rho\  \cdot)=0 \ .
\]
Moreover, $(\rho\frac{d}{d\rho})^k\rho^lt(\rho\  \cdot)|_{\rho=1}=c(t_0)\,\delta$, where $c(t_0)$ is a homogeneous differential operator of order $l-d$ which is independent of the choice of the extension $t$.
\end{teo}
\begin{proof} The proof of the first part of the theorem may be found in the mentioned literature, e.g. in \cite{HW2}. The second statement follows from the fact that different extensions differ by a $(l-d)$th derivative of the $\delta$-function which is homogeneous of degree $l$ and thus does not contribute to $c(t_0)$. 
\end{proof}
We now want to compute the scaling violations of extensions of homogeneous distributions for some typical examples (cf.~e.g.~\cite{GS}).
We first recall fundamental solutions of the Laplacian on the pseudo Riemannian spaces $\R^{d-s,s}$, $s<d$, 
$d$ even and $d>2$, with the metric $g=\mathrm{diagonal}(\underbrace{+1,\ldots,+1}_{d-s},\underbrace{-1,\ldots,-1}_{s})$. 
Let $x^2=\sum_{i=1}^{d-s}(x^i)^2-\sum_{i=d-s+1}^d (x^i)^2$ and 
$\square=\sum_{i=1}^{d-s}\frac{\partial^2}{\partial (x^i)^2}-\sum_{i=d-s+1}^d \frac{\partial^2}{\partial (x^i)^2}$. Then
\begin{lemma}
\[\square \frac{1}{(x^2-i\varepsilon)^{\frac{d}{2}-1}}=i^s(2-d)\,|\mathrm{S}^{d-1}|\,\delta(x)\]
where $|\mathrm{S}^{d-1}|$ is the volume of the unit sphere in $d$ dimensions.
\end{lemma}
\begin{proof}
Let $g$ be any constant Riemannian metric on $\R^d$. Then
\[
\partial_{\mu}\sqrt{\det g}g^{\mu\nu}\partial_{\nu}\frac{1}{(x^\mu x^\nu g_{\mu\nu})^{\frac{d}{2}-1}}=(2-d)\,|\mathrm{S}^{d-1}|\,\delta(x) \ .
\]
Let 
$g_{\mu\nu}=\mathrm{diagonal}(\underbrace{1,\ldots,1}_{d-s},\underbrace{z,\ldots,z}_s)$. The left hand side is an analytic function of $z$ outside of the negative real axis. We then take the limit $z\to -1$ in the lower halfplane and obtain $\sqrt{\det g}=(-i)^s$ and thus
\[
(-i)^s\square \frac{1}{(x^2-i\varepsilon)^{\frac{d}{2}-1}}=(2-d)\,|\mathrm{S}^{d-1}|\,\delta(x) \ .
\] 
\end{proof}
The first example we want to treat is $\frac{1}{(x^2-i\varepsilon)^{\frac{d}{2}}}$. This is a well defined distribution outside of the origin. Let $t$ be an almost homogeneous extension. Then
\[
\rho\frac{d}{d\rho}\rho^d t(\rho x)|_{\rho=1}=\partial_{\mu}(x^{\mu}t) \ .
\]
But $x^{\mu}t$ is an almost homogeneous extension of the distribution $\frac{x^{\mu}}{(x^2-i\varepsilon)^{\frac{d}{2}}}$ which has degree $d-1$, thus  $x^{\mu}t$ is unique. Moreover
\[
\frac{x^{\mu}}{(x^2-i\varepsilon)^{\frac{d}{2}}}=\frac{1}{2-d}\partial^{\mu}\frac{1}{(x^2-i\varepsilon)^{\frac{d}{2}-1}} \ .
\] 
Thus 
\begin{equation}
\rho\frac{d}{d\rho}\rho^d t(\rho x)|_{\rho=1}=\frac{1}{2-d}\square\frac{1}{(x^2-i\varepsilon)^{\frac{d}{2}-1}}=i^s|\mathrm{S}^{d-1}|\delta(x) \ ,\label{example1}
\end{equation}
and we find
\[c\left(\frac{1}{(x^2-i\varepsilon)^{\frac{d}{2}}}\right)=i^s|\mathrm{S}^{d-1}|\]

To derive an explicit expression for an extension $t$ we set $u\doteq (x^2-i\epsilon)$. We have to find a distribution $F(u)$ such that
$\square_x\,F(u)=u^{-\frac{d}{2}}$. Due to
\[
 \square_x\,F(u)=2d\,F'(u)+4u\,F''(u)=\frac{4}{u^{\frac{d}{2}-1}}\,\frac{d}{du}(u^\frac{d}{2}\,F'(u))
\]
the differential equation can easily be integrated:
\[
F'(u)=\frac{\mathrm{log}(-\kappa^2u)-\frac{1}{\frac{d}{2}-1}}{4\,u^\frac{d}{2}}
\]
where $\kappa^2>0$ is an integration constant. A further integration yields $F(u)$ and we get the extension
\[
t(x)=\square_x\,F(u)=\frac{1}{4-2d}\,\square \frac{\log(-\kappa^2(x^2-i\eps))}{(x^2-i\eps)^{\frac{d}{2}-1}}\ .
\]

We now consider the distributions $\frac{1}{(x^2-i\eps)^{\frac{d}{2}+k}}$, $k\in\NN$. Since these distributions are invariant under the pseudo-orthogonal group $\mathrm{O}(d-s,s)$, the differential operator characterizing the scaling violations of extensions must have the form
\[
c\left(\frac{1}{(x^2-i\eps)^{\frac{d}{2}+k}}\right)=c_k \square^k, \ c_k\in\C \ .
\]
Let $t_k$ be an extension of $\frac{1}{(x^2-i\eps)^{\frac{d}{2}+k}}$. We multiply both sides of the equation
\[
\rho\frac{d}{d\rho}\rho^{d+2k}t_k(\rho x)|_{\rho=1}=c_k\square^k\delta(x)
\]
by $(x^2)^k$. We have $(x^2)^k\square^k\delta(x)=\left(\square^k(x^2)^k\right)\delta$.
A straightforward calculation shows 
\[
\square^k (x^2)^k=2^{2k}\frac{k!(\frac{d}{2}+k-1)!}{(\frac{d}{2}-1)!} \ .
\]
On the left hand side the computation can be traced back to (\ref{example1}),
\[
(x^2)^k\rho\frac{d}{d\rho}\rho^{d+2k}t_k(\rho x)|_{\rho=1}=\rho\frac{d}{d\rho}\rho^{d}t(\rho x)|_{\rho=1} =i^s\,|\mathrm{S}^{d-1}|\,\delta(x) \ ,
\]
thus $c_k=i^s\,|\mathrm{S}^{d-1}|\,\frac{(\frac{d}{2}-1)!}{2^{2k}k!(\frac{d}{2}+k-1)!}$.\\
\\
\\
\end{appendix}


{\bf Acknowledgments.}  Part of this work was done at the Erwin-Schr\"odinger Institute, Vienna, during an acticity on new developments in perturbative quantum field theory. In particular we gratefully acknowledge enlightening discussions with Hanno Gottschalk, Christoph Kopper and Manfred Salmhofer on the method of Flow Equations. M. D. also thanks Eberhard Zeidler and the Max Planck Institute for Mathematics in the Sciences, Leipzig, for support.
M. D. was also supported by the German Research Foundation (Deutsche Forschungsgemeinschaft (DFG)) through the Institutional 
Strategy of the University of G\"ottingen.

\bibliographystyle{my-h-elsevier}

\end{document}